\documentclass{article}

\usepackage{amsthm,amsmath,amssymb,sidecap,subfig,graphicx,hyperref}
\usepackage[left=2.5cm,top=3cm,right=3cm]{geometry}
\usepackage[lined,boxed,commentsnumbered]{algorithm2e}
{
\theoremstyle{definition}
\newtheorem{examp}{Example}

}
\newtheorem{lem}{Lemma}
\newtheorem{theo}{Theorem} 
\newtheorem{coro}{Corollary}

\newtheorem{prop}{Proposition}

\newcommand\meetirrdiff[2]{\mathfrak{m}(#1,#2)}
\newcommand\downvertices[1]{Anc(#1)}
\newcommand\downn[2]{e_{#1,#2}}
\newcommand\downnn[2]{{#1}_{#2}}
\newcommand\down[1]{e_{#1}}

\newcommand\successor[1]{{#1}^{+}}
\newcommand\shotvector[2]{sh_{#1}(#2)}
\newcommand\set[1]{\{ #1 \}}
\newcommand\oneshotvector[1]{sh_{#1}}

\newcommand\labelmap[1]{\mathfrak{m}(#1)}

\graphicspath{%
    {converted_graphics/}% inserted by PCTeX
    {D:/PhD/CFGLattices/Images/}% inserted by PCTeX
}

\begin{document}
\title{Lattices generated by Chip Firing Game models: criteria and recognition algorithm\thanks{This paper was partially sponsored by the Vietnamese National Foundation for Science and Technology Development (NAFOSTED)}} 
         % Enter your title between curly braces
%\tnotetext[t1]{This paper was partially sponsored by the Vietnamese National Foundation for Science and Technology Development}
%\author{Trung Van PHAM \corref{cor2}}
%\ead{pvtrung@math.ac.vn}
%
%\author{Thi Ha Duong PHAN}
%\ead{phanhaduong@math.ac.vn}
%
%\address{Vietnam Institute of mathematics,18 Hoang Quoc Viet Road, Cau Giay District, Hanoi, Vietnam}
%
%%\address[hdp]{Vietnam Institute of mathematics, 18 Hoang Quoc Viet, Cau Giay district, Hanoi, Vietnam}
%
%\cortext[cor2]{Corresponding author}

%\maketitle
%\\
%Vietnam Institute of Mathematics\\
%\texttt{\small (pvtrung,phanhaduong)@math.ac.vn}
%}        % Enter your name between curly braces
%%\email{*pvtrung@math.ac.vn \hspace{5.ex}**phanhaduong@math.ac.vn}
%\author{Trung Van Pham, Thi Ha Duong Phan\\
%Department of Mathematics of Computer Science\\
%Vietnamese Institute of Mathematics\\
%18 Hoang Quoc Viet, Cau Giay district, Hanoi, Vietnam
%}
\author{Trung Van Pham and Thi Ha Duong Phan}
\maketitle
\begin{abstract}
It is well-known that the class of lattices generated by Chip Firing games (CFGs) is strictly included in the class of upper locally distributive lattices (ULD). However a necessary and sufficient criterion for this class is still an open question. In this paper we settle this problem by giving such a criterion.  This criterion provides a polynomial-time algorithm for constructing a CFG which generates a given lattice if such a CFG exists. Going further we solve the same problem on two other classes of lattices which are generated by CFGs on the classes of undirected graphs and directed acyclic graphs.\\
\hspace{6.ex}\textit{\textbf{Keywords.}} Abelian Sandpile model, Chip Firing Game, discrete dynamic model, lattice, Sandpile model, ULD lattice, linear programming
\end{abstract}
\section{Introduction}
The Chip Firing Game (CFG) is a discrete dynamical model which was first defined by A. Bj\"orner, L. Lov\'asz and W. Shor while studying the `balancing game' \cite{BL92,BLS91,BTW87,S86}. The model has various applications in many fields of science such as physics \cite{DRSV95,BTW87}, computer science \cite{BL92,BLS91,GMP98}, social science \cite{B97,B99} and mathematics \cite{B99,M97,M01}.

The model is a game which consists of a directed multi-graph $G$ (also called \emph{support graph}), the set of \emph{configurations} on $G$ and an \emph{evolution rule} on this set of configurations. Here, a configuration $c$ on $G$ is a map from the set $V(G)$ of vertices of $G$ to non-negative integers. For each vertex $v$ the integer $c(v)$ is regarded as the number of chips stored in $v$. In a configuration $c$, vertex $v$ is \emph{firable} if $v$ has at least one outgoing edge and $c(v)$ is at least the out-degree of $v$. The \emph{evolution rule} is defined as follows. When $v$ is firable in $c$, $c$ can be transformed into another configuration $c'$ by moving one chip stored in $v$ along each outgoing edge of $v$. 
\begin{SCfigure} % float placement: (h)ere, page (t)op, page (b)ottom, other (p)age
  \centering
  % file name: E:/PhD/ULDlattices2/Images/image34.pdf
  \includegraphics[bb=9 13 466 316,width=2.67in,height=1.77in,keepaspectratio]{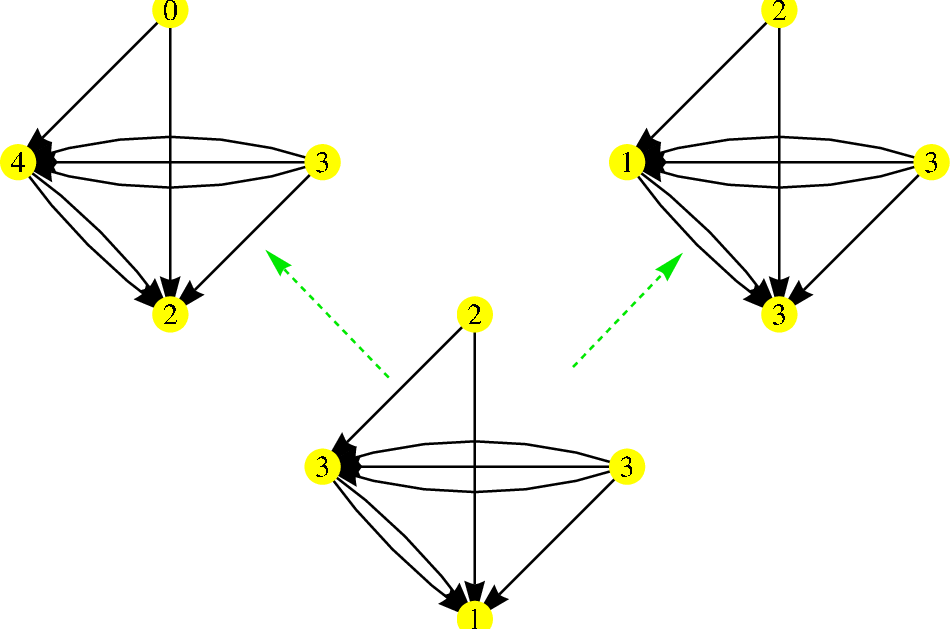}
  \caption{The number at each vertex indicates the number of chips stored there. The configuration at the bottom of the figure can be transformed into two new configurations since it has two firable vertices}
  \label{fig:image34}
\end{SCfigure}
We call this process \emph{firing $v$}, and write $c\overset{v}{\to}c'$. An \emph{execution} is a sequence of firing and is often written in the form $c_1\overset{v_1}{\to}c_2\overset{v_2}{\to}c_3\cdots\to c_{k-1}\overset{v_{k-1}}{\to}c_k$. The set of configurations which can be obtained from $c$ by a sequence of firing is called \emph{configuration space}, and denoted by ${\rm CFG(G,c)}$. 

A CFG begins with an initial configuration $c_0$. It can be played forever or reaches a unique fixed point where no firing is possible \cite{Eri93}. When the game reaches the unique fixed point, ${\rm CFG(G,c_0)}$ is an \emph{upper locally distributive lattice} with the order defined by setting $c_1\leq c_2$ if $c_1$ can be transformed into $c_2$ by a (possibly empty) sequence of firing  \cite{LP01}. A CFG is \emph{simple} if each vertex is fired at most once during any of its executions. Two CFGs are \emph{equivalent} if their generated lattices are isomorphic. Let ${\rm L(CFG)}$ denote the class of lattices generated by CFGs. A well-known result is that $D\subsetneq {\rm L(CFG)\subsetneq ULD}$  \cite{MVP01}, where $D$ and $ULD$ denote the classes of distributive lattices and upper locally distributive lattices, respectively. Despite of the results on inclusion, one knows little about the structure of ${\rm L(CFG)}$, even an algorithm for determining whether a given ULD lattice is in ${\rm L(CFG)}$ is unknown so far.

The Chip Firing Game has some important restrictions. An important restriction is the Abelian Sandpile model (ASM), a restriction of CFGs on undirected graphs  \cite{Mag03,BTW87,BLS91}\footnote{The term ``Abelian Sandpile model'' has been introduced in Dhar's earlier work in which the author focused on studying the algebraic property of the recurrent configurations of the model \cite{Dha90}. The definition for the Abelian Sandpile model we present here follows from the work of Magnien in which the author used the term again with a different definition \cite{Mag03}.}. This model has been extensively studied in recent years. In \cite{Mag03} the author studied the class of lattices generated by ASMs, denoted by ${\rm L(ASM)}$, and showed that this class of lattices is strictly included in ${\rm L(CFG)}$ and strictly includes the class of distributive lattices. As ${\rm L(CFG)}$, the structure of ${\rm L(ASM)}$ is little known. An algorithm for determining whether a given ULD lattice is in ${\rm L(ASM)}$ is still open.

The goal of our study is to find conditions that completely characterize those classes of lattices. One of the most important discoveries in our study is pointing out a strong connection  between the objects which does not seem to be closely related. These objects are meet-irreducibles,  simple CFGs, firing vertices of a CFG, and systems of linear inequalities. In particular, we establish a one-to-one correspondence between the firing vertices of a simple CFG and the meet-irreducibles of the lattice generated by this CFG. Using this correspondence we achieve a necessary and sufficient condition for ${\rm L(CFG)}$. By generalizing this correspondence to CFGs that are not necessarily simple, we also obtain a necessary and sufficient condition for ${\rm L(ASM)}$. Both conditions provide polynomial-time algorithms that address the above computational problems. As an application of these conditions, we present in this paper a lattice in ${\rm L(CFG)\backslash L(ASM)}$ that is smaller than the one shown in \cite{Mag03}.

In \cite{Mag03}, to prove $\rm{D \subsetneq L(ASM)}$ the author studied simple CFGs on directed acyclic graphs (DAGs) and showed that such a CFG is equivalent to a CFG on an undirected graph. It is natural to study CFGs on DAGs which are not necessarily simple. Again our method is applicable to this model and we show that any CFG on a DAG is equivalent to a simple CFG on a DAG. As a corollary, the class of lattices generated by CFGs on DAGs is strictly included in ${\rm L(ASM)}$. We also give a necessary and sufficient condition for the class of lattices generated by this model.

Section \ref{second section} gives some preliminary definitions, notations and results on lattice and Chip Firing games. In Sections \ref{third section}, \ref{fourth section} and \ref{fifth section} we study the properties of three classes of lattices generated by CFGs on general graphs, undirected graphs and directed acyclic graphs, respectively. These sections are devoted to necessary and sufficient criteria for determining which class of lattices a given ULD lattice belongs to. In the conclusion we give some open problems which are currently in our interests.
\section{Preliminary definitions and previous results}
\label{second section}
\subsection{\small \textit{Notations and definitions}}
\hspace{1.ex}Let $L=(X,\leq)$ be a finite partial order ($X$ is equipped with a binary relation $\leq$ which is transitive, reflexive and antisymmetric). A subset $I$ of $X$ is called an \emph{ideal} of $L$ if for every $x \in I$ and $y \in X$ such that $y \leq x$ we have $y \in I$. For $x,y\in X$, $y$ is an \emph{upper cover} of $x$ if $x<y$ and for every $z\in X$, $x \leq z \leq y$ implies that $z=x$ or $z=y$. If $y$ is an upper cover of $x$ then $x$ is a \emph{lower cover} of y, and then we write $x \prec y$. A finite partial order is often presented by a Hasse diagram in which for each cover $x \prec y$ of $L$, there is a curve that goes upward from $x$ to $y$. The lattice $L$ is a \emph{lattice} if any two elements of $L$ have a least upper bound (\emph{join}) and a greatest lower bound (\emph{meet}). When $L$ is lattice, we have the following notations and denitions
\begin{itemize}
  \item $\textbf{0},\textbf{1}$ denote the minimum and the maximum of $L$. 
  \item for every $x,y \in X$, $x\lor y$ and $x \land y$ denote the join and the meet of $x,y$, respectively. 
  \item for $x\in X$, $x$ is a \emph{meet-irreducible} if it has exactly one upper cover. The element $x$ is a \emph{join-irreducible} if $x$ has exactly one lower cover. Let $M$ and $J$ denote the collections of the meet-irreducibles and the join-irreducibles of $L$, respectively. Let $M_x,J_x$ be given by $M_x=\set{m \in M: x \leq m}$ and $J_x=\set{j \in J: j \leq x}$. For $j \in J,m \in M$, if $j$ is a minimal element in $X\backslash  \set{ x \in X: x \leq m }$ then we write $j \downarrow m$. If $m$ is a maximal element in $X \backslash  \set{x \in X: j \leq x}$ then we write $j \uparrow m$, and $j \updownarrow m$ if $j \downarrow m$ and $j \uparrow m$. 
  \item The lattice $L$ is a \emph{distributive lattice} if it satisfies one of the following equivalent conditions
       \begin{itemize}
         \item[1. ] for every $x,y,z\in X$, we have $x \land (y \lor z)=(x \land y) \lor (x \land z)$. 
         \item[2. ] for every $x,y,z\in X$, we have $x \lor (y \land z)=(x \lor y) \land (x \lor z)$.
       \end{itemize}  
For a finite set $A$, $(2^{A},\subseteq)$ is a distributive lattice. A lattice generated in this way is called \emph{hypercube}. 
 \item for $x,y\in X$ satisfying $x \leq y$, $[x,y]$ stands for set $\set{z \in X: x \leq z \leq y}$. If $x \neq \textbf{1}$, $x^{+}$ denotes the join of all upper covers of $x$. Note that if $x$ is a meet-irreducible then $x^{+}$ is the unique upper cover of $x$. If $x\neq \textbf{0}$, $x^{-}$ denotes the meet of all lower covers of $x$. If $x$ is a join-irreducible then $x^{-}$ is the unique lower cover of $x$. The lattice $L$ is an \emph{upper locally distributive (ULD) lattice} \cite{M90,D40} if for every $x \in X$, $x \neq \textbf{1}$ implies the sublattice induced by $[x,x^{+}]$ is a hypercube. By dual notion, $L$ is a \emph{lower locally distributive (LLD) lattice} if for every $x\in X$, $x\neq \textbf{0}$ implies that the sublattice induced by $[x^{-},x]$ is a hypercube.
\end{itemize}  

Let $G$ be a directed multi-graph. For $v_1,v_2\in V$, $E(v_1,v_2)$ denotes the number of edges from $v_1$ to $v_2$. It follows that $E(v_1,v_1)$ is the number of loops at $v_1$. For $v \in V$, the out-degree of $v$, denoted by $deg^{+}(v)$, is defined by $deg^{+}(v)=\underset{v' \in V}{\sum}E(v,v')$ and the in-degree of $v$, denoted by $deg^{-}(v)$, is defined by $\deg^{-}(v) =\underset{v'\in V}{\sum}E(v',v)$. A vertex $v$ of $G$ is called \emph{sink} if it has no outgoing edge, \emph{i.e.} $deg^{+}(v)=E(v,v)$. A subset $C$ of $V(G)$ is a \emph{closed component} if $|C|\geq 2$, $C$ is a strongly connected component and there is no edge going from $C$ to a vertex outside of $C$. A CFG, which is defined on a graph having no closed component, always reaches a unique fixed point, moreover its configuration space is a ULD lattice  \cite{BL92,LP01}. If ${\rm CFG(G,c_0)}$ has a unique fixed point and ${\rm CFG(G,c_0)}$ is isomorphic to a ULD lattice $L$, we say ${\rm CFG(G,c_0)}$ generates $L$. Then we can identify the configurations of ${\rm CFG(G,c_0)}$ with the elements of $L$ (by an isomorphism).

\textbf{Remark. }Throughout this paper when ${\rm CFG(G,c_0)}$ generates $L$, the configurations in ${\rm CFG(G,c_0)}$ are automatically identified with the elements of $L$. All later arguments use this assumption.
\subsection{\small \textit{Previous results}}
\begin{theo}[Birkhoff \cite{B33}]
A lattice is distributive if and only if it is isomorphic to the lattice of the ideals of the order induced by its meet-irreducibles.
\end{theo}
\begin{lem}[Caspard \cite{C98}]
\label{condition on cover relation for ULD lattice}
A lattice $L=(X,\leq)$ is upper locally distributive if and only if for any $x,y \in X$,
$$
x \prec y \Leftrightarrow M_y \subset M_x \text{ and } |M_x\backslash M_y|=1
$$
\end{lem}
\begin{lem}[Latapy and Phan \cite{LP01}]
In a CFG reaching a unique fixed point, if two sequences of firing are starting at the same configuration and leading to the same configuration then for every $v \in V(G)$, the number of times v fired in each sequences are the same, where $G$ is the support on which the game is  defined. 
\end{lem}
In a ${\rm CFG(G,c_0)}$ having a unique fixed point, for each $c$ being a configuration in ${\rm CFG(G,c_0)}$, the \emph{shotvector} of $c$, denoted by $\oneshotvector{c}$, assigns each vertex $v$ of $G$ to the number of times $v$ fired in any execution from the  configuration $c_0$  to $c$. Thus $\oneshotvector{c}$ is a map from $V(G)$ to $\mathbb{N}$. It follows from the above lemma that the shotvector of $c$ is well-defined. For $c_1,c_2\in {\rm CFG(G,c_0)}$ we write $\oneshotvector{c_1} \leq \oneshotvector{c_2}$ if for every $v \in V(G)$, $\shotvector{c_1}{v}\leq \shotvector{c_2}{v}$. It is known that $\oneshotvector{c_1}\leq \oneshotvector{c_2}$ if and only if $c_1$ can be transformed into $c_2$ by a sequence of firing \cite{LP01}. 

Throughout the coming sections, we always work with a general finite ULD lattice $L=(X,\leq)$. Recall that $M,J$ denote the collections of the meet-irreducibles and the join-irreducibles of $L$, respectively. The map $\mathfrak{m}:\set{(x,y): x \prec y \text{ holds in } L}\to M$ is given by $\mathfrak{m}(x,y)$ is the element in $M_x\backslash M_y$. All graphs are supposed to be directed multi-graphs. In a CFG if configuration $c$ can be transformed into $c'$ by firing some vertex in the support graph then we denote this unique vertex by $\vartheta(c,c')$. All CFGs, which are considered in this paper, are assumed to be reaching a fixed point. To denote a CFG, a configuration space and a lattice generated by a CFG, we will use the common notation ${\rm CFG(G,c_0)}$ since all of them are completely defined by $G$ and $c_0$.
\section{A necessary and sufficient condition for ${\rm L(CFG)}$}
\label{third section}
Given a ULD lattice $L$, is $L$ in ${\rm L(CFG)}$? This question was asked in \cite{MVP01}. Up to now, there exists no good criterion for  ${\rm L(CFG)}$ that suggests a polynomial-time algorithm for this computational problem. In this section we address this problem by giving a necessary and sufficient condition for ${\rm L(CFG)}$. We recall an important result in \cite{MVP01}
\begin{theo}[Magnien, Vuillon and Phan \cite{MVP01}]
\label{theorem of simple CFG}
Any CFG that reaches a unique fixed point is equivalent to a simple CFG
\end{theo}
From now until the end of this section, all CFGs are supposed to be simple. The following lemma is known in \cite{FK09}. Since it will play an important role in this paper and its proof is simple, it is presented here with a proof.
\begin{lem}[Felsner and Knauer \cite{FK09}]
\label{lemma of the square connection}
Let $a,b$ be two elements of $L$ such that $a \prec b$. Let $m$ denote $\mathfrak{m}(a,b)$. Then for any chain $a=x_1 \prec x_2 \prec \cdots \prec x_k=m$ in $L$, there exists a chain $b=y_1 \prec y_2 \prec \cdots \prec y_k=\successor{m}$ in $L$ such that $x_i \prec y_i$ for every $1\leq i \leq k$. Moreover, $\mathfrak{m}(x_i,y_i)=m$  for every $1\leq i\leq k$. 
\end{lem}
\begin{proof}
%\begin{center}
%\includegraphics[bb=112 17 268 262,width=1.19in,height=1.88in,keepaspectratio]{image9}\\
%\end{center}
It's clear that $x_2 \neq y_1$. Since $L$ is a ULD lattice, there exists a unique $y_2$ such that $y_1 \prec y_2$ and $x_2 \prec y_2$. It follows easily that $\labelmap{x_1,y_1}=\labelmap{x_2,y_2}=m$. If $k=2$ then $y_2=\successor{m}$. Otherwise  repeat the previous argument starting with $x_2,y_2$ until the index reaches $k$. We obtain the sequence $b=y_1 \prec y_2 \prec \cdots \prec y_k=\successor{m}$ which has the desired property.
\end{proof}
\begin{figure}[!h]
\centering
\includegraphics[bb=112 17 268 262,width=1.19in,height=1.88in,keepaspectratio]{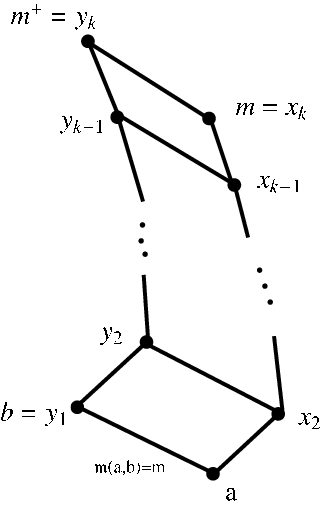}
\end{figure}
\begin{lem}
\label{lemma of relation}
Let $L$ be a ULD lattice generated by ${\rm CFG(G,c_0)}$ and let $\mathcal{V}$ denote the set of vertices which are fired in ${\rm CFG(G,c_0)}$. For each $c \in {\rm CFG(G,c_0)}$, $\vartheta(c)$ denotes the set of vertices which are fired to obtain $c$.   Then
\begin{itemize}
  \item[1. ]  The map $\kappa: M\to \mathcal{V}$ determined by $\forall m \in M, \kappa(m)=\vartheta(c,c')$, where $c,c'$ are two elements in $L$ such that $c \prec c'$ and $\meetirrdiff{c}{c'}=m$, is well-defined. Furthermore $\kappa$ is a bijection.
  \item[2. ] For every $c\in {\rm CFG(G,c_0)}$, $\vartheta(c)=\kappa(M\backslash  M_{c}) $.
\end{itemize}  
\end{lem}
\begin{proof}$\text{ }$\\
\begin{itemize}
\item[1. ] The map $\kappa$ is defined on whole $M$ since for every $m\in M$, $\mathfrak{m}(m,m^{+})=m$. To prove $\kappa$ is well-defined, it suffices to show that for each $m\in M$, if $\labelmap{a,b}=m$ then $\vartheta(a,b)=\vartheta(m,m^{+})$.  Let $a=x_1 \prec x_2 \prec \cdots \prec x_k=m$. By Lemma \ref{lemma of the square connection}, there exists $b=y_1 \prec y_2 \cdots \prec y_k=\successor{m}$ such that for every $1 \leq i \leq k$, we have $ x_i \prec y_i$. Therefore $\vartheta (a,b)=\vartheta (x_1,y_1)=\vartheta (x_2,y_2)=\cdots=\vartheta (x_k,y_k)=\vartheta (m,\successor{m})$.

\hspace{3.ex} Clearly $\kappa$ is surjective. To prove $\kappa$ is bijective, it suffices to show that $|M|=|\mathcal{V}|$. Let $\textbf{0}=c_0\overset{v_1}{\to}c_1 \overset{v_2}{\to}c_2\overset{v_3}{\to}\cdots \to c_{N-1}\overset{v_N}{\to}c_N=\textbf{1}$ be an execution to obtain the fixed point. Since $M_\textbf{0}=M,M_\textbf{1}=\emptyset$ and for every $0 \leq i \leq N-1$, $|M_{c_i}\backslash M_{c_{i+1}} |=1$, it follows that $N=|M|$, therefore $|\mathcal{V}|=|M|$
\item[2. ] Let $\textbf{0}=d_0\overset{v_1}{\to}d_1\overset{v_2}{\to}d_2\overset{v_3}{\to}\cdots\overset{v_k}{\to}d_k=c\overset{v_{k+1}}{\to}d_{k+1}\to\cdots \overset{d_{N-1}}{\to} d_{N-1}\overset{v_N}{\to}d_N=\textbf{1}$ be an execution to obtain the fixed point. It's clear that $\vartheta(c)=\set{v_1,v_2,\cdots,v_k}$, therefore $\vartheta(c)=\set{v_1,v_2,\cdots,v_N}\backslash \{v_{k+1},v_{k+2},\cdots,$ \linebreak $v_N\}$. By the definition of $\kappa$, we have $\kappa(M)=\set{v_1,v_2,\cdots,v_N}$ and $\{v_{k+1},v_{k+2},\cdots,v_N\}=\set{\kappa(\labelmap{d_i,d_{i+1}}):k \leq i \leq N-1}=\kappa(\set{\labelmap{d_i,d_{i+1}}: k \leq i\leq N-1})=\kappa\big(\underset{k \leq i \leq N-1}{\bigcup} M_{d_i}\backslash M_{d_{i+1}}\big)=\kappa(M_{d_{k}}\backslash M_{d_N})=\kappa(M_{d_{k}})=\kappa(M_c)$. It follows that $\vartheta(c)=\kappa(M)\backslash \kappa(M_c)=\kappa(M\backslash M_c)$ since $\kappa$ is bijective. 
\end{itemize}
\end{proof}

The lemma means that if $L$ is generated by a CFG then each meet irreducible of $L$ can be considered as a vertex of its support graph. It is an important point to set up a criterion for ${\rm L(CFG)}$. For better understanding, we give an example for this correspondence. The CFG which is defined on the support graph and the initial configuration shown in Figure \ref{fig:image38} and Figure \ref{fig:image39}
\begin{figure}
\centering
\subfloat[Support graph]{\label{fig:image38}\includegraphics[bb=1 4 134 140,width=1.19in,height=1.22in,keepaspectratio]{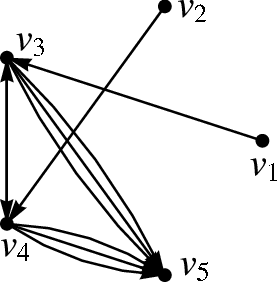}}
\quad
\subfloat[Initial configuration]{\label{fig:image39}\includegraphics[bb=9 11 189 199,width=1.27in,height=1.33in,keepaspectratio]{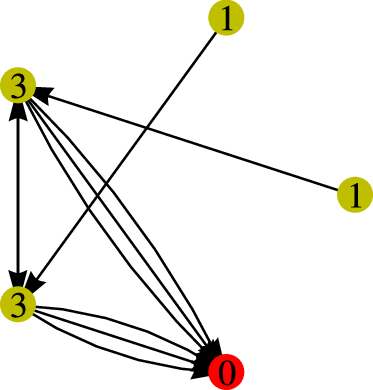}}
\caption{An example of Chip firing game}
\end{figure}
\noindent generates the lattice represented in Figure \ref{fig:image16}.\begin{figure}
\centering
\subfloat[Cover relation labeled by firing vertices]{\label{fig:image16}\includegraphics[bb=13 7 214 268,width=1.45in,height=1.88in,keepaspectratio]{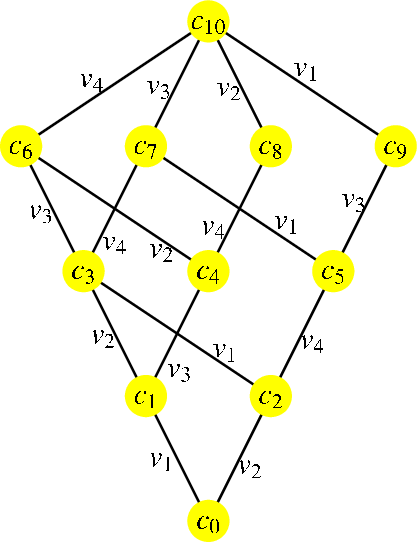}}
\quad
\subfloat[Configurations labeled by fired vertices]{\label{fig:image19}\includegraphics[bb=4 0 249 292,width=1.58in,height=1.88in,keepaspectratio]{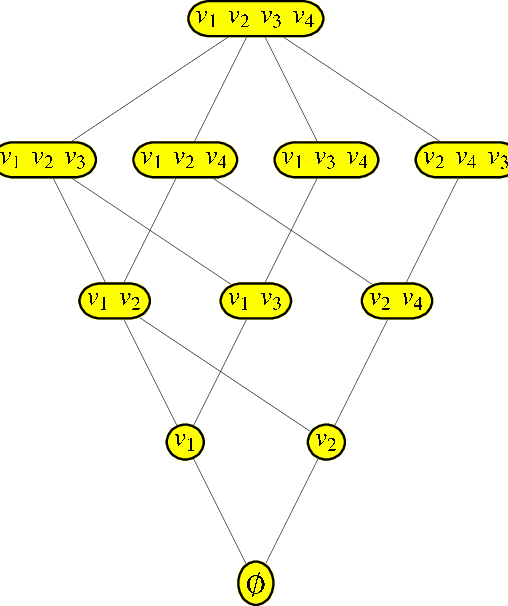}}
\caption{firing-vertex labeling}
\end{figure}
\begin{figure}
\centering
\subfloat[Cover relation labeled by meet-irreducibles]{\label{fig:image17}\includegraphics[bb=13 11 222 270,width=1.52in,height=1.88in,keepaspectratio]{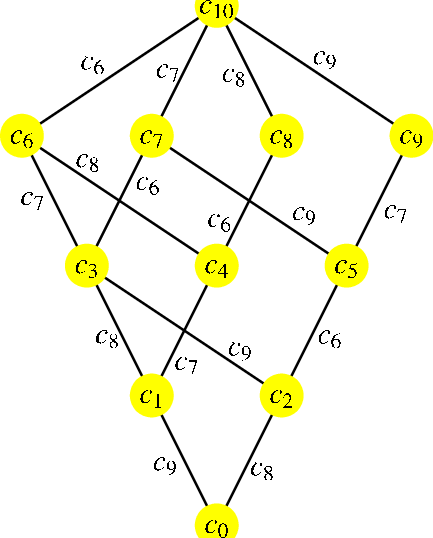}}
\quad
\subfloat[Configurations labeled by meet-irreducibles]{\label{fig:image18}\includegraphics[bb=4 0 232 271,width=1.58in,height=1.88in,keepaspectratio]{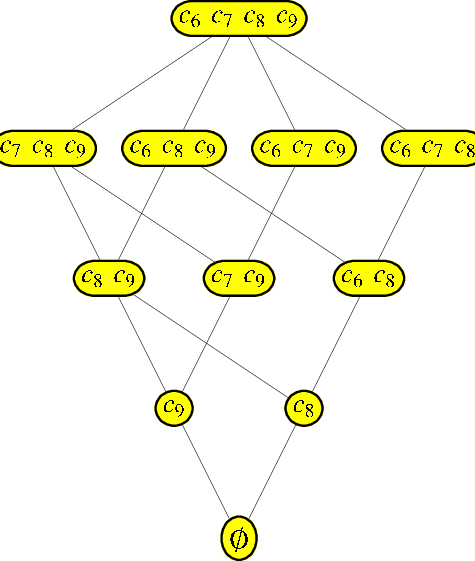}}
\caption{meet-irreducible labeling}
\end{figure}
\noindent In Figure \ref{fig:image16}, each $c_i \prec c_j$ is labeled by the vertex which is fired in $c_i$ to obtain $c_j$. The lattice in Figure \ref{fig:image17} is the same as one in Figure \ref{fig:image16} but each $c_i \prec c_j$ is labeled by $\mathfrak{m}(c_i,c_j)$. Figure \ref{fig:image19} shows the lattice in the way each configuration is presented by the set of vertices which are fired to obtain this configuration. In Figure \ref{fig:image18}, each configuration $c$ is presented by $M\backslash M_c$.  Clearly the labelings in Figure \ref{fig:image16} and Figure \ref{fig:image17} are the same, the presentations in Figure \ref{fig:image19} and Figure \ref{fig:image18} are the same too with respect to the correspondence $\kappa$ defined by $\kappa(c_{6})=v_4,\kappa(c_{7})=v_3,\kappa(c_{8})=v_2,\kappa({c_9})=v_1$.

For each $m \in M$, $\mathfrak{U}_m$ denotes the collection of all minimal elements of $\set{x \in X: \exists y \in X, x \prec y \text{ and } \mathfrak{m}(x,y)=m}$ and $\mathfrak{L}_m$ denotes the collection of all maximal elements of $X\backslash \underset{a \in \mathfrak{U}_m}{\bigcup} \set{x \in X: a\leq x}$. 

Let us explain in a few words why we define these notations. Suppose that $L$ is generated by ${\rm CFG(G,c_0)}$. For a vertex $v$ fired in the game, we consider all configurations in ${\rm CFG(G,c_0)}$ which have enough chips stored at $v$ in order that $v$ can be fired. If we only care about the firability of $v$, we only need to consider the collection $\mathfrak{U}_v$ of all  minimal configurations of these configurations. The configurations, which are not greater than equal to any configuration in $\mathfrak{U}_v$, do not have enough chips stored at $v$ in order that $v$ can be fired. We only need to consider the collection $\mathfrak{L}_v$ of all maximal configurations of these configurations to know the firability of $v$. Sets $\mathfrak{U}_{v},\mathfrak{L}_v$ are exactly $\mathfrak{U}_{\kappa^{-1}(v)}$, $\mathfrak{L}_{\kappa^{-1}(v)}$, respectively in $L$. Note that $\mathfrak{U}_m,\mathfrak{L}_m$ only depend on $L$, not depend on the CFGs generating $L$, even if there exists no such CFGs. 

For each $m \in M$ the system of linear inequalities $\mathcal{E}(m)$ is given by
$$
\mathcal{E}(m)=
\begin{cases}
\set{w-\underset{x \in M\backslash M_a}{\sum} \down{x}\geq 1: a \in \mathfrak{L}_m}\bigcup \set{w\leq \underset{x \in M\backslash M_a}{\sum}\down{x} : a \in \mathfrak{U}_m}&\text{ if }\mathfrak{U}_m\neq \set{\textbf{0}}\\
\set{w\geq 1}&\text{ if }\mathfrak{U}_m=\set{\textbf{0}}
\end{cases}
$$
where $w$ is an added variable. The collection of all variables of $\mathcal{E}(m)$ is $\set{w}\cup \set{\down{x}: x \in \underset{a\in \mathfrak{U}_m\cup \mathfrak{L}_m}{\bigcup}(M\backslash M_a)}$. It follows from the definitions of $\mathfrak{U}_m$ and $\mathfrak{L}_m$ that if $\down{x}$ is in the collection of all variables of $\mathcal{E}(m)$ then $x \neq m$. Note that $\mathcal{E}(m)=\set{w \geq 1}$ if and only if there exists $x \in X$ such that $\textbf{0} \prec x$ and $\mathfrak{m}(\textbf{0},x)=m$.

\textbf{Remark.} When $L$ is generated by some CFG, Lemma \ref{lemma of relation} means that each $m \in M$ can be regarded as a vertex of this CFG. The system of linear inequalities $\mathcal{E}(m)$ describes the firability of $m$ in the following meaning. In order that $m$ can be fired, $m$ receives at least $w$ chips from its neighbors. Each $e_x$ in $\mathcal{E}(m)$ indicates the number of chips that $x$ sends to $m$ when it is fired. For each $a \in \mathfrak{U}_m\cup \mathfrak{L}_m$ when all vertices in $M\backslash M_a$ are fired, the game arrives at the configuration $a$, and $m$ receives $\underset{x \in M\backslash M_a}{\sum} e_x$ chips from its neighbors. The vertex $m$ is not firable in each $a \in \mathfrak{L}_m$, therefore $w-\underset{x \in M\backslash M_a}{\sum}e_x\geq 1$. Similarly $m$ is firable in each $ a \in \mathfrak{U}_m$, therefore $w \leq \underset{x \in M\backslash M_a}{\sum} e_x$. 
\begin{examp}
\label{first example}
We consider again the lattice presented in Figure \ref{fig:image17}.  We have $M=\set{c_6,c_7, c_{8}, c_{9}}, \mathfrak{U}_{c_8}=\mathfrak{U}_{c_9}=\set{c_0} ,\mathfrak{U}_{c_6}= \set{c_2,c_4}, \mathfrak{U}_{c_7}=\set{c_1,c_5}, \mathfrak{L}_{c_8}=\mathfrak{L}_{c_{9}}=\emptyset,\mathfrak{L}_{c_6}=\set{c_1},\mathfrak{L}_{c_7}=\set{c_2}$ . Then 
$$
\mathcal{E}(c_8)=\mathcal{E}(c_9)=\set{w\geq 1}
$$
$$
\mathcal{E}(c_6)=\left\{
\begin{array}{c}
w \leq \down{c_8}\\
w \leq \down{c_7}+\down{c_9}\\
w-\down{c_9}\geq 1
\end{array}
\right. ;
\mathcal{E}(c_7)=\left\{
\begin{array}{c}
w \leq \down{c_9}\\
w\leq \down{c_6}+\down{c_8}\\
w-\down{c_8}\geq 1
\end{array}
\right. 
$$
\end{examp}

\begin{lem}
\label{lemma of solution}
If $L \in {\rm L(CFG)}$ then for every $m \in M$,  $\mathcal{E}(m)$ has non-negative integral solutions. 
\end{lem}
\begin{proof}
%The lemma clearly holds if $\mathcal{E}(m)=\set{w\geq 1}$. If $\mathcal{E}(m)\neq \set{w \geq 1}$, there exists a CFG, say ${\rm CFG(G,c_0)}$, which generates L. We can identify the elements of $L$ with the configurations of ${\rm CFG(G,c_0)}$. Without loss of generality, we can assume that all vertices of $G$ are fired exactly once in the game except for only one vertex $s$. Note that ${\rm CFG(G,c_0)}$ remains unchanged if we remove all out-edges of $s$, therefore $s$ can be considered as the sink of the game. 

Let ${\rm CFG(G,c_0)}$ be a CFG that generates $L$. If $\mathcal{E}(m)=\set{w \geq 1}$ then clearly it has a non-negative integral solution. Otherwise, let $ f_m: \set{\down{x}: x \in \underset{a\in \mathfrak{L}_m \cup \mathfrak{U}_m}{\bigcup} (M\backslash M_a)}\cup \set{w}\to \mathbb{N}$ be given by
 $$
f_m(y)=\begin{cases}
E(\kappa(x),\kappa(m))&\text{ if } y=\down{x} \text{ for some } x \in \underset{a \in \mathfrak{U}_m\cup \mathfrak{L}_m}{\bigcup} (M\backslash M_a) \\
deg^{+}(\kappa(m))-c_0(\kappa(m))&\text{ if } y=w
\end{cases}$$
where $\kappa$ is the map which is defined as in Lemma \ref{lemma of relation}. Note that since $\mathcal{E}(m)\neq \set{w\geq 1}$, $\kappa(m)$ cannot be fired at the beginning of the game, therefore $\deg^{+}(\kappa(m))-c_0(\kappa(m))>0$.

We show that $f_m$ is a solution of $\mathcal{E}(m)$. Indeed let $a\in \mathfrak{U}_m$. By Lemma \ref{lemma of relation} the set of vertices which are fired to obtain $a$ is $\kappa(M\backslash M_a)$. After firing all vertices in $\kappa(M\backslash M_a)$ $\kappa(m)$ receives $\underset{x\in M\backslash M_a}{\sum} E(\kappa(x),\kappa(m))$ chips from its neighbors. Since $\kappa(m)$ is firable in $a$, it follows that $\underset{x\in M\backslash M_a}{\sum} f_m(\down{x})= \underset{x\in M\backslash M_a}{\sum} E(\kappa(x),\kappa(m))\geq deg^{+}(\kappa(m))-c_0(\kappa(m))=f_m(w)$. It remains to prove that for $a'\in \mathfrak{L}_m$, we have $\underset{x\in M\backslash M_{a'}}{\sum}f(\down{x})<f(w)$. It follows from the definition of $\mathfrak{L}_m$ and from Lemma \ref{lemma of relation} that $\kappa(m)$ is not firable in $a'$. By a similar argument we have $\underset{x\in M\backslash M_{a'}}{\sum}f_m (\down{x})=\underset{x\in M\backslash M_{a'}}{\sum}E(\kappa(x),\kappa(m))<deg^{+}(\kappa(m))-c_0(\kappa(m))=f_m(w)$.
\end{proof}
\begin{theo}
\label{the condition of lattices induced by CFG}
$L$ is in ${\rm L(CFG)}$ if and only if for each $m \in M$, $\mathcal{E}(m)$ has non-negative integral solutions.
\end{theo}
\begin{proof}
$\Rightarrow$ has been proved by Lemma \ref{lemma of solution}. It remains to show that $\Leftarrow$ is also true. We are going to construct a graph $G$ and an initial configuration $c_0$ so that the game is simple and ${\rm CFG(G,c_0)}$ is isomorphic to $L$.

The set of vertices of $G$ is $M\cup \set{s}$, where $s$ is distinct from $M$ and will play a role as the sink of $G$. The edges of $G$ are constructed as follows. For each $m\in M$ let $f_m:U_m\to \mathbb{N}$ be a solution of $\mathcal{E}(m)$, where $U_m$ is the collection of all variables in $\mathcal{E}(m)$. Set $E(m,s)=f_m (w)+\underset{v \in M \text{ and }e_v\in U_m \backslash \set{w}}{\sum} f_m (e_v)$ and for each $v\in M$ satisfying $e_v \in U_m\backslash \set{w}$, and $E(v,m)=f_m(e_v)$.

The constructing graph $G$ has the following properties. The graph $G$ is connected and has no closed component since each vertex $v\neq s$ has at least one edge going from $v$ to $s$, and $s$ has no outgoing edge. Thus any CFG on $G$ reaches a fixed point. For each $m \in V(G)\backslash \set{s}$ we have $deg^{-}(m)=\underset{v \in M\text{ and } e_v\in U_m\backslash\set{w}}{\sum}f_m(e_v)<E(m,s)\leq deg^{+}(m)$ and $deg^{+}(m)=f_m(w)+\underset{v \in M\text{ and } e_v \in U_m\backslash\set{w}}{\sum}f_m(e_v)+\underset{m'\in M\backslash \set{m}}{\sum} f_{m'}(e_m)$. Note that in the formula of $deg^{+}(m)$ if $e_m \not \in U_{m'}\backslash \set{w}$ then we set $f_{m'}(e_m)=0$. The in-degree and the out-degree at each  vertex of $G$ depend on the non-negative integral solutions $f_m$ we choose, therefore they may be large. In fact the number of vertices of $G$ is small, that is $|M|+1$, whereas the number of edges of $G$ is often very large. However this is not a problem of presenting $G$ since a multi-graph is often represented by associating each pair $(v,v')$ of vertices of $G$ with a number that indicates the number of edges from $v$ to $v'$.
% we can do a compact representation of $G$ as follows. We associate each pair $(v,v')$ of vertices of $G$ with a number that indicates the number of edges connecting $v$ to $v'$.

We construct $c_0:V(G)\to\mathbb{N}$ as follows
$$
 c_0(v)=\begin{cases}
deg^{+}(v)-f_v(w)&\text{ if } v \neq s \text{ and } deg^{-}(v)\neq 0\\
deg^{+}(v)&\text{ if }deg^{-}(v)=0\\
0& \text{ if } v=s
\end{cases}
$$

We claim that ${\rm CFG(G,c_0)}$ is simple. Indeed for the sake of contradiction we suppose that there exists at least one vertex in $G$ which is fired more than once in an execution, say $c_0 \overset{v_1}{\to}c_1\overset{v_2}{\to}c_2\cdots\overset{v_{k-1}}{\to}c_{k-1}\overset{v_k}{\to}c_k$, to reach the fixed point  of ${\rm CFG(G,c_0)}$. By the assumption, $v_1,v_2,\cdots,v_k$ are not pairwise distinct. Let $i$ be the largest index such that $v_1,v_2,\cdots,v_{i-1}$ are pairwise distinct. Vertex $v_i$, therefore, is in $\set{v_1,v_2,\cdots,v_{i-1}}$. We have $deg^{-}(v_i)\neq 0$ since $v_i$ is fired more than once during the execution. To obtain $c_{i-1}$ each vertex in $v_1,v_2,\cdots,v_{i-1}$ is fired exactly once, therefore $c_{i-1}(v_i)\leq c_0(v_i)+deg^{-}(v_i)-deg^{+}(v_i)$. Since $v_i$ is firable in $c_{i-1}$, it follows that $deg^{+}(v_i)\leq c_{i-1}(v_i)\leq c_0 (v_i)+deg^{-}(v_i)-deg^{+}(v_i)=-f_{v_i}(w)+deg^{-}(v_i)$. It contradicts the fact that $deg^{+}(v_i)>deg^{-}(v_i)$.

We claim that for every execution $c_0\overset{v_1}{\to}c_1\overset{v_2}{\to}c_2\overset{v_3}{\to}\cdots\overset{v_{k-1}}{\to}c_{k-1}\overset{v_k}{\to}c_k$ of ${\rm CFG(G,c_0)}$, there exists a chain $\textbf{0}=d_0\prec d_1\prec d_2\prec \cdots \prec d_{k-1} \prec d_k$ in $L$ such that $\mathfrak{m}(d_{i-1},d_{i})=v_i$ for every $1 \leq i\leq k$. Note that if the chain exists then it is defined uniquely. We prove the claim by induction on $k$. For $k=1$, $v_1$ is firable in $c_0$. It follows from the construction of $G$ and $c_0$ that only the vertices in $G$ having indegree 0 are firable in $c_0$, therefore $\mathfrak{U}_{v_1}=\set{\textbf{0}}$. It implies that there exists $d_1 \in X$ such that $d_0 \prec d_1$ and $M_{d_0}\backslash M_{d_1}=\set{v_1}$. The claim holds for $k=1$. For $k\geq 2$, let $\textbf{0}\prec d_1\prec d_2\prec \cdots \prec d_{k-1}$ be the chain in $L$ such that $M_{d_{i-1}}\backslash M_{d_i}=\set{v_i}$ for every $1 \leq i \leq k-1$. If $d_{k-1}$ is not less than or equal to any element in $\mathfrak{U}_{v_k}$ then there exists $a \in \mathfrak{L}_{v_k}$ such that $d_{k-1} \leq a$. It follows from the definition of $\mathcal{E}(v_k)$ that $\overset{k-1}{\underset{i=1 }{\sum}}f_{v_k}(\down{v_i})\leq \underset{x \in M\backslash M_a}{\sum}f_{v_k}(\down{x})<f_{v_k}(w)$. It implies that after $v_1,v_2,\cdots,v_{k-1}$ have been fired, $v_k$ receives less than $f_{v_k}(w)$ chips from its neighbors, therefore $v_k$ is not firable in $c_{k-1}$. It's a contradiction. If there exists $b \in \mathfrak{U}_{v_k}$ such that $b\leq d_{k-1}$ then there exists $b' \in X$ such that $b \prec b'$ and $M_b\backslash M_{b'}=\set{v_k}$. Let $d_k=b'\lor d_{k-1}$. It suffices to show that $M_{d_{k-1}}\backslash M_{d_k}=\set{v_k}$. Indeed, we have $M_{d_k}=M_{b'}\cap M_{d_{k-1}}=(M_b \backslash \set{v_k})\cap M_{d_{k-1}}=(M_b\cap M_{d_{k-1}})\backslash \set{v_k}=M_{d_{k-1}}\backslash \set{v_k}$. Since $M\backslash M_{d_{k-1}}=\set{v_1,v_2,\cdots,v_{k-1}}$, $v_k \in M_{d_{k-1}}$, therefore $M_{d_{k-1}}\backslash M_{d_k}=\set{v_k}$.
% Since $v_k \in M_{e_{k-1}}$, it follows that $M_{e_{k-1}}\backslash M_{e_k}=\set{v_k}$.

Our next claim is that for any chain $\textbf{0}=d_0\prec d_1\prec d_2 \prec \cdots \prec d_{k-1}\prec d_k$ in $L$, there exists an execution $c_0 \overset{v_1}{\to}c_1\overset{v_2}{\to}c_2\overset{v_3}{\to}\cdots \overset{v_{k-1}}{\to}c_{k-1}\overset{v_k}{\to}c_k$ in ${\rm CFG(G,c_0)}$, where $v_i=\mathfrak{m}(d_{i-1},d_i)$ for every $1 \leq i \leq k$. We prove the claim by induction on $k$. For $k=1$, we have $\mathfrak{U}_{v_1}=\set{\textbf{0}}$. It follows easily that $deg^{-}(v_1)=0$, therefore $v_1$ is firable in $c_0$. By firing $v_1$ in $c_0$, we obtain $c_0 \overset{v_1}{\to}c_1$. The claim holds for $k=1$. For $k\geq 2$, let $c_0\overset{v_1}{\to}c_1\overset{v_2}{\to}c_2\overset{v_3}{\to}\cdots \overset{v_{k-1}}{\to}c_{k-1}$ be an execution in the game such that $\set{v_i}=M_{d_{i-1}}\backslash M_{d_i}$ for every $1 \leq i\leq k-1$.
Since $d_{k-1}$ is in the set $\set{x \in X: \exists y\in X,x\prec y \text{ and } \mathfrak{m}(x,y)=v_k}$ and $\mathfrak{U}_{v_k}$ is the collection of all minimal elements in this set,  there is $a\in \mathfrak{U}_{v_k}$ such that $a \leq d_{k-1}$. Thus $M\backslash M_a\subseteq \set{v_1,v_2,\cdots,v_{k-1}}$. The vertex $v_k$ receives at least $\underset{x\in M\backslash M_a}{\sum}f_{v_k}(\down{x})$ chips from its neighbors after all vertices $v_1,v_2,\cdots,v_{k-1}$ have been fired. The vertex $v_k$ is firable in $c_{k-1}$ since $\underset{x \in M\backslash M_a}{\sum}f_{v_k}(\down{x})\geq f_{v_k}(w)$. The claim follows.

It follows immediately from the above claims that ${\rm CFG(G,c_0)}$ and $L$ are isomorphic. This completes the proof.
\end{proof}
We establish a relation between $\mathfrak{U}_m$ and the join-irreducibles of $L$. The main result of \cite{MVP01} will follow easily from this relation.
\begin{prop}
\label{relation with join-irreducibles}
For each meet-irreducible $m$ of $L$, $\mathfrak{U}_m=\set{j^{-}: j \in J \text{ and } j \downarrow m}$
\end{prop}
\begin{proof}
For each $m\in M$, let $\mathcal{F}_m$ be given by $\mathcal{F}_m=\set{x \in X: \exists y \in X, x \prec y \text{ and }\meetirrdiff{x}{y}=m}$.  Let $A$ denote $\set{j^{-}: j \in J \text{ and } j\downarrow m}$. First, we show that $A \subseteq \mathfrak{U}_m$. To this end, we prove that $j^{-}\in \mathfrak{U}_m$ for every $j \in J$ satisfying $j \downarrow m$.  Since $j \not \leq m$ and $j^{-} \leq m$, we have $ m \not \in M_{j}$ and $m \in M_{j^{-}}$, therefore $m \in M_{j^{-}}\backslash M_{j}$. Since $|M_{j^{-}}\backslash M_{j}|=1$, it follows that $M_{j^{-}}\backslash M_{j}=\set{m}$, hence $j^{-} \in \mathcal{F}_m$. It remains to prove that $j^{-}$  is a minimal element of $\mathcal{F}_m$. For a contradiction, we suppose that there exists $a \prec b$ in $L$ such that $\meetirrdiff{a}{b}=m$ and $a<j^{-}$. It follows easily that $b < j$, therefore there is a chain $b=d_1\prec d_2\prec \cdots \prec d_k=j$ in $L$ of length $\geq 1$. We have $\meetirrdiff{d_{k-1}}{j}\neq m$ since $\meetirrdiff{a}{b}=m$, therefore $d_{k-1}\neq j^{-}$. It contradicts the fact that $j$ is a join-irreducible. 

We are left with showing that $\mathfrak{U}_m \subseteq A$. Let $a \in \mathfrak{U}_m$. There is a unique element $b$ in $L$ such that $a \prec b$ and $\meetirrdiff{a}{b}=m$. It suffices to show that $b \in J$. For a contradiction, we suppose that $b \not \in J$. Then there exists $c \in X$ such that $c \prec b$ and $c \neq a$. Let $d$ denote the infimum of $a$ and $c$. There exists $a'\in L$ such that $d \prec a'$ and $a' \leq c$. Since $a \in \mathfrak{U}_m$, we have $\meetirrdiff{d}{a'}\neq m$, therefore $m \in M_{a'}$. It follows from $a' \leq b$ that $M_a=M_b \cup \set{m} \subseteq M_{a'} \cup \set{m}=M_{a'}$, hence $a' \leq a$. It contradicts the fact that $d$ is the infimum of $a$ and $c$.
\end{proof}
\begin{coro}
\label{uniqueness of join-irreducibles}
If $L$ is a distributive lattice then for every meet-irreducible $m$ of $L$, we have $|\mathfrak{U}_m|=1$. 
\end{coro}
\begin{proof}
For each $m \in M$, we define $m_{\downarrow}=\set{j \in J: j \downarrow m}$. Note that $m_{\downarrow}\neq \emptyset$. For every $m_1,m_2\in M$, $m_1\neq m_2$ implies that ${m_1}_{\downarrow}\cap {m_2}_{\downarrow} =\emptyset$  since if $j \in m_{\downarrow}$ then $ \meetirrdiff{j^{-}}{j}=m$. In a distributive lattice, the cardinality of the meet-irreducibles is equal to the cardinality of the join-irreducibles, \emph{i.e} $|M|=|J|$. It follows easily that for every $m\in M$, $|m_{{\downarrow}}|=1$, therefore $|\mathfrak{U}_m|=1$ by Proposition \ref{relation with join-irreducibles}.
\end{proof}
\begin{prop}
\label{lemma of the condition on distributive lattice}
If $L$ is a distributive lattice then for each $m \in M$, $\mathcal{E}(m)$ has non-negative integral solutions.
\end{prop}
\begin{proof}
It follows from Corollary \ref{uniqueness of join-irreducibles} that $|\mathfrak{U}_m|=1$. If $\mathfrak{U}_m=\set{\textbf{0}}$ then $\mathcal{E}(m)=\set{w\geq 1}$, therefore the proposition holds. If $\mathfrak{U}_m\neq \set{\textbf{0}}$,  let $u$ denote the unique element in $\mathfrak{U}_m$. The system $\mathcal{E}(m)$ of linear inequalities now becomes
$$
\mathcal{E}(m)=\set{w-\underset{x \in M\backslash M_a}{\sum} \down{x}\geq 1 : a \in \mathfrak{L}_m}\bigcup \set{ w \leq \underset{x \in M\backslash M_u}{\sum}\down{x} }
$$
The collection $U_m$ of all variables in $\mathcal{E}(m)$ is $\set{w}\cup \set{\down{x}: x \in (M\backslash M_u)\cup \underset{a \in \mathfrak{L}_m}{\bigcup} (M\backslash M_a)}$. Let $f:U_m\to \mathbb{N}$ be given by
$$
f(y)=\begin{cases}
1&\text{ if }y=\down{x} \text{ for some } x \in M\backslash M_u \\
|M\backslash M_u|&\text{ if } y=w\\
0& otherwise
\end{cases}
$$

We claim that $f$ is a solution of $\mathcal{E}(m)$. By the definition of $f$, it is straightforward to verify that $f(w)\leq \underset{x\in M\backslash M_u}{\sum} f(\down{x})$. It remains to show that for every $a\in \mathfrak{L}_m$, $\underset{x \in M\backslash M_a}{\sum}f(\down{x})<f(w)$. Indeed, $\underset{x \in M\backslash M_a}{\sum}f(\down{x})=|(M\backslash M_a)\cap (M\backslash M_u)|$ follows from the definition of $f$. Since $u \not \leq a$, we have $(M\backslash M_a)\cap (M\backslash M_u)\subsetneq M \backslash M_u$, therefore $\underset{x \in M\backslash M_a}{\sum}f(\down{x})<|M\backslash M_u|=f(w)$.  
\end{proof}
We derive easily the following corollary from Theorem \ref{the condition of lattices induced by CFG} and Proposition \ref{lemma of the condition on distributive lattice}
\begin{coro}[\cite{MVP01}]%{\cite{MVP01}}
Every distributive lattice is in ${\rm L(CFG)}$.
\end{coro}
We close this section by presenting a polynomial time algorithm for determining whether a given ULD lattice is in ${\rm L(CFG)}$. To do this we have to show that finding a non-negative integral solution of $\mathcal{E}(m)$ can be done in polynomial time. It is well-known that the problem of deciding whether an integral system of linear inequalities has an integral solution is NP-complete. Fortunately the following shows that the problem is solvable in polynomial time when it is restricted to $\mathcal{E}(m)$.
\begin{lem}
\label{find a solution}
Given $\mathcal{E}(m)$, we can decide if it has a non-negative integral solution, and if so, find one, in polynomial time.
\end{lem}
\begin{proof}
Clearly the corresponding problem on $\mathbb{R}$ is solvable in polynomial time by using the known algorithms for linear programming. If $\mathcal{E}(m)$ has no non-negative real solution then $\mathcal{E}(m)$ has no non-negative integral solution. Otherwise let $f'$ be a non-negative real solution of $\mathcal{E}(m)$. We are going to construct a non-negative integral solution $f$ of $\mathcal{E}(m)$ from $f'$. The map $f:\set{w}\cup \set{
e_x: x \in \underset{a \in \mathfrak{L}_m\cup \mathfrak{U}_m}{\bigcup} (M\backslash M_a)
}\to \mathbb{N}$ is defined by $f(e_x)=\lfloor 2|M|f'(e_x)\rfloor$ for every $x \in \underset{a \in \mathfrak{L}_m\cup \mathfrak{U}_m}{\bigcup} (M\backslash M_a)$, and $f(w)=min\set{
\underset{x \in M\backslash M_a}{\sum} f(e_x): a\in \mathfrak{U}_m
}$. We show that $f$ is a non-negative integral solution of $\mathcal{E}(m)$. By the definition of $f(w)$ it remains to show that $\underset{x \in M\backslash M_a}{\sum} f(e_x)<f(w)$ for any $a \in \mathfrak{L}_m$. Let $b \in \mathfrak{U}_m$ such that $f(w)=\underset{x \in M\backslash M_b}{\sum} f(e_x)$. Clearly $\underset{x \in M\backslash M_a}{\sum} f(e_x)\leq 2|M|\underset{x \in M\backslash M_a}{\sum}f'(e_x)\leq 2|M| (f'(w)-1)\leq -2|M|+\underset{x \in M\backslash M_b}{\sum} (2|M|f'(e_x))\leq -2|M|+\underset{x \in M\backslash M_b}{\sum} (\lfloor 2|M|f'(e_x)\rfloor+1)\leq -2 |M|+|M|+f(w)<f(w)$.
\end{proof}

The lattice $L$ can be input as a directed acyclic graph with the edges induced from the cover relation of $L$, \emph{i.e. } $(x,y) \in E(L)$ iff $x \prec y$ holds in $L$. Note that $\mathfrak{U}_m$ and $\mathfrak{L}_m$ can be found in $O(|E(L)|)$ time by using search algorithms.
The algorithm is presented by the following pseudocode\\
\begin{algorithm}[H]
\SetKwInOut{Input}{Input}
\SetKwInOut{Output}{Output}

\Input{A ULD lattice $L$ which is input as a acyclic graph with the edges defined by the cover relation}
\Output{\textbf{Yes} if $L$ is in ${\rm L(CFG)}$, \textbf{No} otherwise. If \textbf{Yes} then give a support graph $G$ and an initial configuration $c_0$ on $G$ so that ${\rm CFG(G,c_0)}$ is isomorphic to $L$}
$V(G):=M\cup \set{s}$\;
$E(G):=\emptyset$\;
\For{$m\in M$}{
Construct $\mathcal{E}(m)$ \;
\lIf{$\mathcal{E}(m)$\text{ has no non-negative integral solutions }}{\textbf{Reject}}\;
\Else{
        Let $f_m$ be a non-negative integral solution of $\mathcal{E}(m)$\;
        Let $U_m$ be the collection of all variables in $\mathcal{E}(m)$\;
        \For{$\down{x}\in U_m\backslash \set{w}$}{
                      Add $f_m (\down{x})$ edges $(x,m)$ to $G$
                     }
        Add $f_m (w)+\underset{\down{x}\in U_m \backslash\set{w}}{\sum}f_m (\down{x})$ edges $(m,s)$ to $G$
        }
}
Construct the initial configuration $c_0$ by
$$
c_0(v):=\begin{cases}
deg^{+}(v)&\text{ if }deg^{-}(v)=0\\
deg^{+}(v)-f_v(w)&\text{ if }deg^{-}(v)\neq 0 \text{ and } v\neq s\\
0&\text{ if }v=s
\end{cases}
$$
\end{algorithm}
We can use the Karmarkar's algorithm \cite{K84} to find a non-negative integral solutions $f_m$ of $\mathcal{E}(m)$ that can be done as in the proof Lemma \ref{find a solution}. For each $m \in M$ the number of bits that are input to the algorithm is bounded by $O(|M|\times |X|)$. We have to run the Karmarkar's algorithm $|M|$ times. Hence the algorithm can be implemented to run in $O(|M|^{6.5}\times {|X|^2}\times log |X|\times log(log |X|))$ time.
\section{A necessary and sufficient condition for ${\rm L(ASM)}$}
\label{fourth section}
Abelian Sandpile model is the CFG model which is defined on connected undirected graphs \cite{BTW87}. In this model, the support graph is undirected and it has a distinguished vertex which is called \emph{sink} and never fires in the game even if it has enough chips. If we replace each undirected edge $(v_1,v_2)$ in the support graph by two directed edges $(v_1,v_2)$ and $(v_2,v_1)$ and remove all out-edges of the sink then we obtain an CFG on directed graph which has the same behavior as the old one. For example, a CFG defined on the following undirected graph with sink $s$
\begin{center}
\includegraphics[bb=0 1 222 219,width=1.47in,height=1.44in,keepaspectratio]{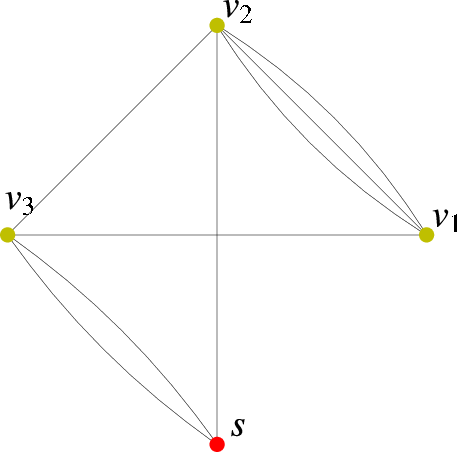}
\end{center}
is the same as one which is defined on the following graph
\begin{center}
\includegraphics[bb=1 3 242 254,width=1.38in,height=1.44in,keepaspectratio]{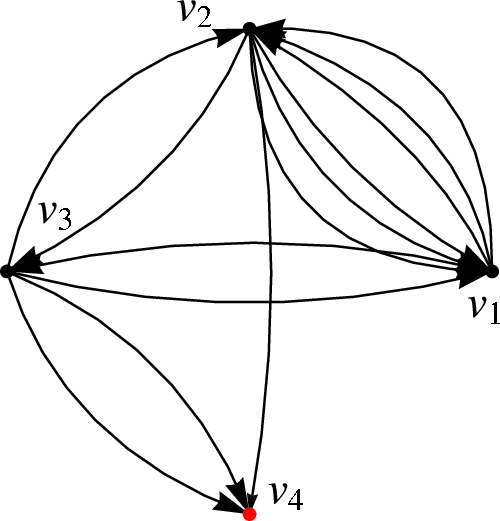}
\end{center}
and the initial configuration is the same as the old one. Thus a ${\rm ASM}$ can be regarded as a CFG on a directed multi-graph. We give an alternative definition of ${\rm ASM}$ on directed multi-graphs as follows. A ${\rm CFG(G,c_0)}$, where $G$ is a directed multi-graph, is a ${\rm ASM}$ if $G$ is connected, $G$ has only one sink $s$ and for any two distinct vertices $v_1,v_2$ of $G$, which are distinct from the sink, we have $E(v_1,v_2)=E(v_2,v_1)$. Therefore in this model we will continue to work on directed multi-graphs. 

The lattice structure of this model was studied in \cite{Mag03}. The authors proved that the class of lattices induced by ASMs is strictly included in ${\rm L(CFG)}$ and strictly includes the class of distributive lattices. To get the necessary and sufficient condition for ${\rm L(CFG)}$, we used the important result from \cite{MVP01} which asserts that every CFG is equivalent to a simple CFG.  A difficulty of getting a necessary and sufficient criterion for ${\rm L(ASM)}$ is that we do not know whether a similar assertion holds for the ${\rm ASM}$, {\emph i.e.} whether an ASM is equivalent to a simple ASM, therefore the argument in \cite{MVP01} does not seem to be transferable to ASM. Nevertheless, we overcome this difficulty by constructing a generalized correspondence between the firing vertices in a relation with their times of firing of a CFG  and the meet-irreducibles of the lattice generated by this CFG. Using this correspondence we achieve a necessary and sufficient condition for ${\rm L(ASM)}$. This condition provides a polynomial-time algorithm for determining whether a given ULD lattice is in ${\rm L(ASM)}$. We also give some other results which concern to this model. The following lemma shows that correspondence, it is a generalization of Lemma \ref{lemma of relation}.
\begin{lem}
\label{generalization of lemma of relation}
If $L$ is generated by ${\rm CFG(G,c_0)}$ then the map $\kappa: M \to V(G)\times \mathbb{N}$, determined by $\kappa(m)=(\vartheta(c,c'),\shotvector{c'}{\vartheta(c,c')})$, where $c,c'$ are two configurations of ${\rm CFG(G,c_0)}$ such that $c\prec c'$ and $\meetirrdiff{c}{c'}=m$, is well-defined. Furthermore $\kappa$ is injective.  
\end{lem}
Note that games in Lemma \ref{lemma of relation} are supposed to be simple, whereas games in the above lemma are not necessarily simple. The lemma means that if each $c\prec c'$ is labeled by the pair of the vertex at which $c$ is fired to obtain $c'$ and the number of times this vertex is fired to reach $c'$ from the initial configuration then the labeling is the same as  labeling $c \prec c'$ by $\mathfrak{m}(c,c')$. Let us give a concrete example to illustrate this concept. The CFG defined by the support graph $G$ and the initial configuration $c_0$, which are shown in Figure \ref{fig:image4546}, generates the lattice that is shown by Figure \ref{fig:image22} and Figure \ref{fig:image23}.
\begin{figure}[!h]
\centering
\subfloat[Support graph]{\label{fig:image45}\includegraphics[bb=2 2 258 204,width=1.55in,height=1.22in,keepaspectratio]{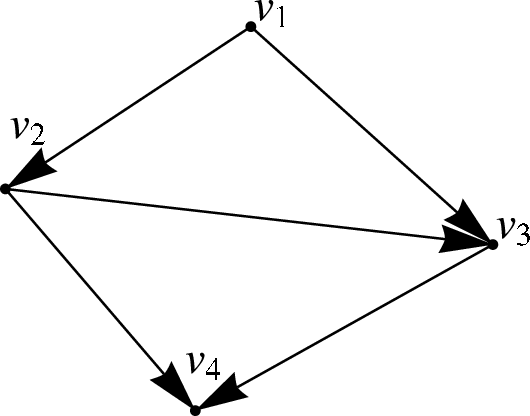}}
\qquad \qquad
\subfloat[Initial configuration]{\label{fig:image46}\includegraphics[bb=10 14 182 150,width=1.55in,height=1.22in,keepaspectratio]{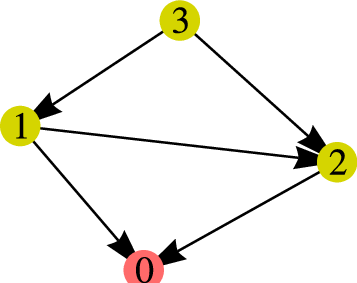}}
\caption{A non-simple CFG}
\label{fig:image4546}
\end{figure}
\captionsetup[subfloat]{justification=centerfirst,singlelinecheck=false}
\begin{figure}[!h]
\centering
\subfloat[Cover relation labeled by firing vertex and times of firing]{\label{fig:image22}\includegraphics[bb=12 6 159 394,width=0.962in,height=2.54in,keepaspectratio]{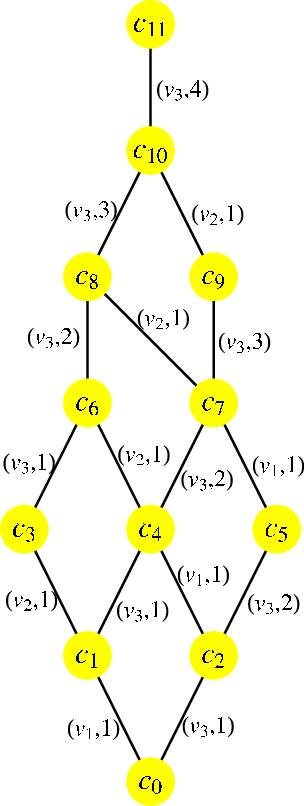}}
\qquad \qquad
\subfloat[Cover relation labeled by meet-irreducibles]{\label{fig:image23}\includegraphics[bb=13 4 174 433,width=0.951in,height=2.54in,keepaspectratio]{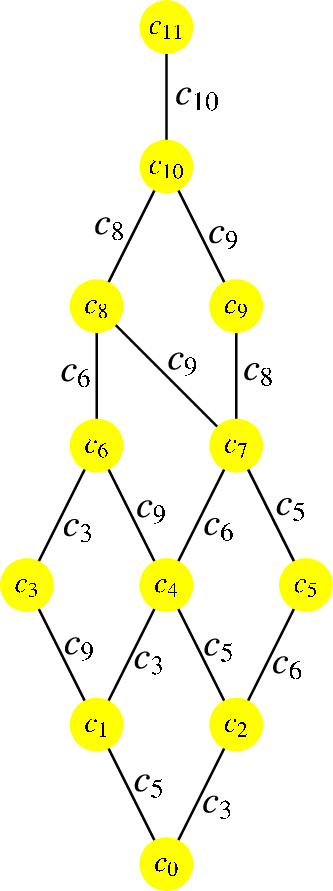}}
\caption{Two ways of labeling}
\end{figure}
\captionsetup[subfloat]{justification=justified,singlelinecheck=true}In Figure \ref{fig:image22}, each $c\prec c'$ is labeled by the fired vertex and the number of times this vertex is fired to obtain $c'$. Figure \ref{fig:image23} shows the lattice in the way each $c\prec c'$ is labeled by $\meetirrdiff{c}{c'}$. It is obvious that the labelings are the same with respect to the correspondence $c_3\to (v_3,1),c_5 \to (v_1,1),c_6\to (v_3,2),c_8\to (v_3,3),c_9\to (v_2,1),c_{10}\to (v_3,4)$.\\

\noindent\emph{{Proof of Lemma \ref{generalization of lemma of relation}}}. To prove $\kappa$ is well-defined, it suffices to show that for $c,c'$ being two configurations of ${\rm CFG(G,c_0)}$ such that $\meetirrdiff{c}{c'}=m$, where $m\in M$, we have $(\vartheta(c,c'),\shotvector{c'}{\vartheta(c,c')})=(\vartheta(m,m^{+}),\shotvector{m^{+}}{\vartheta(m,$ \linebreak $m^{+})})$. Let $c=c_1\prec c_2\prec c_3\prec \cdots \prec c_k=m$ be an execution in $L$. By Lemma \ref{lemma of the square connection}, there exists a chain $c'=d_1\prec d_2\prec d_3\prec \cdots \prec d_k=m^{+}$ such that $c_i \prec d_i$ for every $1 \leq i \leq k$. It is easy to see that $\vartheta(c,c')=\vartheta(c_1,d_1)=\vartheta(c_2,d_2)=\vartheta(c_3,d_3)=\cdots=\vartheta(c_k,d_k)=\vartheta(m,m^{+})$. Let $v$ denote $\vartheta(c,c')$. It remains to prove that $\shotvector{c'}{v}=\shotvector{m^{+}}{v}$. For each $1\leq i\leq k-1$, we have $v=\vartheta(c_i,d_i)\neq \vartheta(c_i,c_{i+1})$, therefore $\shotvector{d_{i+1}}{v}=1+\shotvector{c_{i+1}}{v}=1+\shotvector{c_i}{v}=\shotvector{d_i}{v}$. It implies that $\shotvector{c'}{v}=\shotvector{d_1}{v}=\shotvector{d_2}{v}=\cdots=\shotvector{d_k}{v}=\shotvector{m^{+}}{v}$.

It follows immediately from the definition of $\kappa$ that $\kappa$ is a surjection from $M$ to $\underset{v \in V(G)}{\bigcup} \big(\set{v}\times [\shotvector{\textbf{1}}{v}]\big)$. Here $[n]$ denotes the set $\set{1,2,\cdots,n}$, by convention, $[n]=\emptyset$ if $n \leq 0$. For $v\in V(G)$, $|\set{v}\times [\shotvector{\textbf{1}}{v}]|$ is the number of times $v$ is fired in any execution from $\textbf{0}$ to $\textbf{1}$, therefore $|\underset{v \in V(G)}{\bigcup} \big(\set{v}\times [\shotvector{\textbf{1}}{v}]\big)|$ is the number of times the vertices are fired to reach $\textbf{1}$. Thus $|\underset{v \in V(G)}{\bigcup} \big(\set{v}\times [\shotvector{\textbf{1}}{v}]\big)|$ is the height of $L$. To prove $\kappa$ is injective it suffices to show that $|M|$ is also the height of $L$. Let $\textbf{0}=c_0\prec c_1\prec c_2\prec \cdots\prec c_{h-1}\prec c_h=\textbf{1}$  be any chain of $L$, where $h$ is the height of $L$. Clearly $M=M_{\textbf{0}}\backslash M_{\textbf{1}}=\underset{0 \leq i \leq h-1}{\bigcup} M_{c_i}\backslash M_{c_{i+1}}$, and $(M_{c_i}\backslash M_{c_{i+1}})\cap (M_{c_{i'}}\backslash M_{c_{i'+1}})=\emptyset$ for any two distinct indices $i,i'$ in $\set{0,1,2,\cdots,h-1}$. By Lemma \ref{condition on cover relation for ULD lattice} each $M_{c_i}\backslash M_{c_{i+1}}$ contains exactly one element, therefore $|M|=h$.\hfill $\square$

In the case of directed graphs the systems $\mathcal{E}(m)$ of linear inequalities are solved independently to know whether $L$ is in ${\rm L(CFG)}$ since there is no requirement for relation between $E(v_1,v_2)$ and $E(v_2,v_1)$ on support graph. In the case of undirected graphs the condition $E(v_1,v_2)=E(v_2,v_1)$ must be satisfied for any two vertices distinct from sink. Hence the systems of linear inequalities for ${\rm ASM}$ are constructed as follows.
% This property no longer holds for the case of undirected graphs since in an undirected graph $E(v_1,v_2)=E(v_2,v_1)$ for any two vertices $v_1,v_2$ distinct from sink. So we construct the systems of linear inequalities as follows.

For each $\mathcal{E}(m)$ we define the system of linear inequalities $\mathfrak{E}(m)$ by replacing each variable $\down{x}$ in $\mathcal{E}(m)$ by $\downn{x}{m}$ and $w$ by $\downnn{w}{m}$. We give an example for this transformation. Consider the lattice shown in Figure \ref{fig:image17}. We have $$\mathcal{E}(c_8)=\mathcal{E}(c_9)=\set{w\geq 1}$$
$$
\mathcal{E}(c_6)=\left\{
\begin{array}{l}
w \leq \down{c_8}\\
w \leq \down{c_7}+\down{c_9}\\
w-\down{c_9}\geq 1
\end{array}
\right. ;
\mathcal{E}(c_7)=\left\{
\begin{array}{l}
w \leq \down{c_9}\\
w\leq \down{c_6}+\down{c_8}\\
w-\down{c_8}\geq 1
\end{array}
\right. 
$$
then 
$$
\mathfrak{E}(c_8)=\set{\downnn{w}{c_8}\geq 1};
\mathfrak{E}(c_9)=\set{\downnn{w}{c_9}\geq 1}
$$
$$
\mathfrak{E}(c_6)=\left\{\begin{array}{l}
\downnn{w}{c_6}\leq \downn{c_8}{c_6}\\
\downnn{w}{c_6}\leq \downn{c_7}{c_6}+\downn{c_9}{c_6}\\
\downnn{w}{c_6}-\downn{c_9}{c_6}\geq 1
\end{array}\right. ;
\mathfrak{E}(c_7)=\left\{\begin{array}{l}
\downnn{w}{c_7}\leq \downn{c_9}{c_7}\\
\downnn{w}{c_7}\leq \downn{c_6}{c_7}+\downn{c_8}{c_7}\\
\downnn{w}{c_7}-\downn{c_8}{c_7}\geq 1
\end{array}\right.
$$

For each $m\in M$, $\mathfrak{E}(m)$ is a system of linear inequalities whose variables are a subset of $\set{\downn{m_1}{m_2}: m_1\in M, m_2 \in M\text{ and } m_1 \neq m_2}\cup \set{\downnn{w}{m}: m\in M}$. Let $U$ denote the set of all variables in $\underset{m \in M}{\bigcup} \mathfrak{E}(m)$. The system $\Omega$ of linear inequalities is given by $$\Omega=\left(\underset{m \in M}{\bigcup}\mathfrak{E}(m)\right)\cup \set{\downn{m_1}{m_2}=\downn{m_2}{m_1}: \downn{m_1}{m_2}\text{ and }\downn{m_2}{m_1} \text{ both are in }U }$$ 

If $L$ is generated by a simple CFG, say ${\rm CFG(G,c_0)}$, then it follows from the correspondence established in Lemma \ref{lemma of relation} and the construction in Theorem \ref{the condition of lattices induced by CFG} that for $m_1,m_2\in M$, $\downn{m_1}{m_2}$ can be regarded as the number of directed edges from $v_1$ to $v_2$ in $G$, where $v_1,v_2$ are the corresponding vertices of $m_1,m_2$, respectively. As the sufficient condition in Theorem \ref{the condition of lattices induced by CFG}, the following lemma shows a similar assertion for ${\rm L(ASM)}$.
\begin{lem}
\label{sufficient condition for L(ASM)}
If $\Omega$ has non-negative integral solutions then $L \in {\rm L(ASM)}$.
\end{lem}
\begin{proof}
We construct the graph $G$ whose set of vertices is $M\cup \set{s}$ and the edges are defined as follows. Let $f:U\to \mathbb{N}$ be a non-negative integral solution of $\Omega$. For each two distinct elements $m_1,m_2\in M$, if $\downn{m_1}{m_2}\in U$ then there are $f(\downn{m_1}{m_2})$ edges connecting $m_1$ to $m_2$ in $G$ and $f(\downn{m_1}{m_2})$ edges connecting $m_2$ to $m_1$.  If $\downn{m_1}{m_2}\not \in U$ and $\downn{m_2}{m_1}\not \in U$ then there is no edge connecting $m_1$ with $m_2$ in $G$. It follows immediately from the definition of $\Omega$ that $G$ is well-defined. For each $m \in M$, there are $f(\downnn{w}{m})+\underset{m' \in M \text{ and } m' \neq m}{\sum} E(m',m)$ edges connecting $m$ to $s$. The initial configuration $c_0:V(G)\to \mathbb{N}$ for the game is defined by 
$$
c_0(v)=\begin{cases}
deg^{+}(v)& \text{ if } v \in M \text{ and } \mathfrak{U}_v= \set{\textbf{0}}\\  
deg^{+}(v)-f(\downnn{w}{v})& \text{ if } v \in M \text{ and } \mathfrak{U}_v\neq \set{\textbf{0}}\\
0&\text{ if } v=s
\end{cases}
$$
where $deg^{+}(v)$ denotes the out-degree of $v$ in $G$. It's clear that $s$ is the sink of the game. We can argue similarly as in the proof of Theorem \ref{the condition of lattices induced by CFG} that the game is simple and generates $L$. 
\end{proof}
\begin{examp}
\label{second example ?}
We consider the system of linear inequalities of Example \ref{first example}. Then $\Omega$ is the following system
$$
\Omega=\left\{\begin{array}{l}
\downnn{w}{c_8}\geq 1\\
\downnn{w}{c_9}\geq 1\\
\downnn{w}{c_6}\leq \downn{c_8}{c_6}\\
\downnn{w}{c_6}\leq \downn{c_7}{c_6}+\downn{c_9}{c_6}\\
\downnn{w}{c_6}-\downn{c_9}{c_6}\geq 1\\
\downnn{w}{c_7}\leq \downn{c_9}{c_7}\\
\downnn{w}{c_7}\leq \downn{c_6}{c_7}+\downn{c_8}{c_7}\\
\downnn{w}{c_7}-\downn{c_8}{c_7}\geq 1\\
\downn{c_6}{c_7}=\downn{c_7}{c_6}
\end{array}\right.
$$
The map $f:U\to \mathbb{N}$ defined by 
$$
f(x)=\begin{cases}
0& \text{ if } x=\downn{c_9}{c_6} \text{ or } x=\downn{c_8}{c_7}\\
1& otherwise
\end{cases}
$$
is a non-negative integral solution of $\Omega$. By the construction in the sketch of proof, $G$ and the initial configuration are presented by Figure \ref{fig:image2526}.
\begin{figure}[h]
\centering
%\parbox{1.2in}{
%\includegraphics[bb=0 1 224 231,width=1.08in,height=1.11in,keepaspectratio]{image25}
%\caption*{Support graph}
%}
%\qquad \qquad
%\parbox{1.2in}{
%\includegraphics[bb=11 9 234 242,width=1.08in,height=1.11in,keepaspectratio]{image26}
%\caption*{Initial configuration}
%}
\subfloat[Support graph]{\label{fig:image25}\includegraphics[bb=0 1 224 231,width=1.08in,height=1.11in,keepaspectratio]{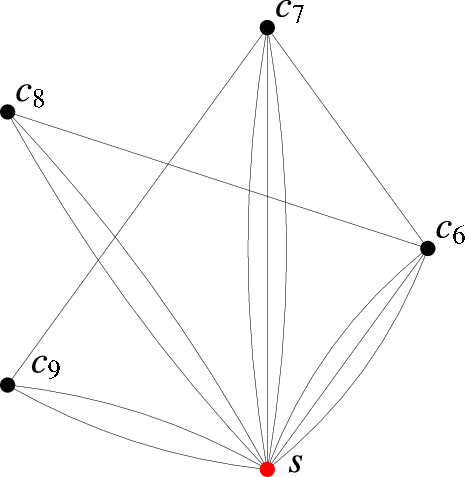}}
\qquad \qquad
\subfloat[Initial configuration]{\label{fig:image26}\includegraphics[bb=11 9 234 242,width=1.08in,height=1.11in,keepaspectratio]{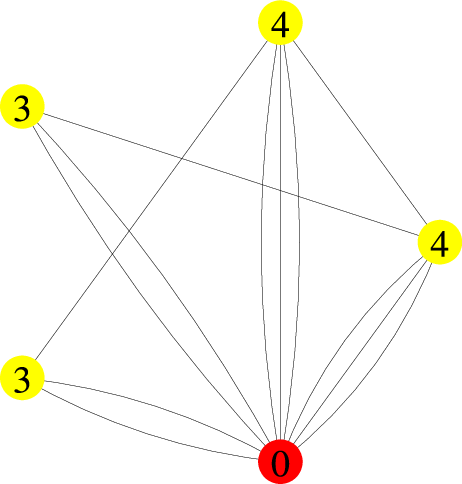}}
\caption{A ASM solution}
\label{fig:image2526}
\end{figure}
Note that in the figure, $G$ is presented by an undirected graph for a nice presentation. The sink of the game is in black. Doing a simple computation on the game, it is straightforward to verify that the lattice generated by ${\rm CFG(G,c_0)}$ is isomorphic to $L$. 

The following theorem shows that the condition that $\Omega$ has non-negative integral solutions is not only a sufficient condition but also a necessary condition of $L \in {\rm L(ASM)}$. 
\end{examp}
\begin{theo}
\label{necessary and sufficient condition for ASM}
$L\in {\rm L(ASM)}$ if and only if $\Omega$ has non-negative integral solutions.
\end{theo}
\begin{proof}
The direction $\Leftarrow$ is already proved by Lemma \ref{sufficient condition for L(ASM)}. We are left with proving the direction $\Rightarrow$. Let ${\rm CFG(G,c_0)}$ be a ${\rm ASM}$ and generates $L$. Let $s$ denote the sink of the game. We define $$N=1+\underset{v \in V(G)\text{ and } v\neq s}{\sum}(\shotvector{\textbf{1}}{v}\times deg^{+}(v))$$
\noindent where $deg^{+}(v)$ denotes the out-degree of $v$ in $G$. Let $\kappa:M\to V(G)\times \mathbb{N}$ be the injective map which is defined as in Lemma \ref{generalization of lemma of relation}. For $m \in M$, let $\kappa(m)^{(1)},\kappa(m)^{(2)}$ denote the first and second components of $\kappa(m)$, respectively. 

We claim that for each $m \in M$ and each $a \in \mathfrak{L}_m$, we have $|\set{(v,n) \in \kappa(M\backslash M_a): v=\kappa(m)^{(1)}} |\leq \kappa(m)^{(2)}-1$. Indeed, let $\textbf{0}=c_0\prec c_1\prec\cdots \prec c_{k-1}\prec c_k=a$ be a chain of $L$, where $k$ is a non-negative integer. Note that $\kappa(M\backslash M_a)=\set{\kappa(\mathfrak{m}(c_i,c_{i+1})):0 \leq i \leq k-1}$. For a contradiction, we suppose that $|\set{(v,n)\in \kappa(M\backslash M_a): v=\kappa(m)^{(1)}}|\geq \kappa(m)^{(2)}$. It implies that the number of times $\kappa(m)^{(1)}$ is fired in the execution (chain) is greater than or equal to $\kappa(m)^{(2)}$. Hence there is a unique index $0\leq j \leq k-1$ such that $\vartheta(c_j,c_{j+1})=\kappa(m)^{(1)}$ and $|\set{i:i \leq j \text{ and }\vartheta(c_i,c_{i+1})=\kappa(m)^{(1)} }|=\kappa(m)^{(2)}$. It follows from the definition of $\kappa$ in Lemma \ref{generalization of lemma of relation} that $\kappa(\mathfrak{m}(c_j,c_{j+1}))=\kappa(m)$. Since $\kappa$ is injective, it follows that $\mathfrak{m}(c_j,c_{j+1})=m$. It contradicts the definition of $\mathfrak{L}_m$. The claim follows. 

Our next claim is that for each $m \in M$ and each $a \in \mathfrak{L}_m$, if $|\set{(v,n)\in \kappa(M\backslash M_a):v=\kappa(m)^{(1)}}|=\kappa(m)^{(2)}-1$ then for every $b \in \mathfrak{U}_m$, we have
\begin{equation}
\label{first equa}
\underset{\begin{array}{c} x\in M\backslash M_a\\ \kappa(x)^{(1)}\neq \kappa(m)^{(1)}  \end{array}}{\sum}E(\kappa(x)^{(1)},\kappa(m)^{(1)})< \underset{\begin{array}{c}x\in M\backslash M_b \\ \kappa(x)^{(1)}\neq \kappa(m)^{(1)}  \end{array}}{\sum}E(\kappa(x)^{(1)},\kappa(m)^{(1)})
\end{equation}
Indeed, the righ-hand side of \eqref{first equa} indicates the number of chips vertex $\kappa(m)^{(1)}$ receives from its neighbors during an execution from $\textbf{0}$ to $b$. To reach $b$, $\kappa(m)^{(1)}$ has been fired $\kappa(m)^{(2)}-1$ times. It follows that the number of chips stored at $\kappa(m)^{(1)}$ in $b$ is
\begin{align}
\label{second equa}
\notag c_0(\kappa(m)^{(1)})&-(\kappa(m)^{(2)}-1)\times \big(deg^{+}(\kappa(m)^{(1)})-E(\kappa(m)^{(1)},\kappa(m)^{(1)})\big)+\\
&+\underset{\begin{array}{c}x\in M\backslash M_b \\ \kappa(x)^{(1)}\neq \kappa(m)^{(1)}  \end{array}}{\sum}E(\kappa(x)^{(1)},\kappa(m)^{(1)})
\end{align}
$\kappa(m)^{(1)}$ is firable in $b$, therefore $\eqref{second equa}\geq deg^{+}(\kappa(m)^{(1)})$. By a similar argument, the number of chips stored at $\kappa(m)^{(1)}$ in $a$ is
\begin{align}
\label{third equa}
\notag c_0(\kappa(m)^{(1)})&-(\kappa(m)^{(2)}-1)\times \big(deg^{+}(\kappa(m)^{(1)})-E(\kappa(m)^{(1)},\kappa(m)^{(1)})\big)+\\
&+\underset{\begin{array}{c} x\in M\backslash M_a\\ \kappa(x)^{(1)}\neq \kappa(m)^{(1)}  \end{array}}{\sum}E(\kappa(x)^{(1)},\kappa(m)^{(1)})
\end{align}
$\kappa(m)^{(1)}$ is not firable in $a$, therefore $\eqref{third equa}<deg^{+}(\kappa(m)^{(1)})$. It follows easily from $\eqref{second equa}$ and $\eqref{third equa}$ that $\eqref{first equa}$ holds. 

Let $f:U\to \mathbb{N}$ be given by
$$
f(y)=\begin{cases}
E(\kappa(m_1)^{(1)},\kappa(m_2)^{(1)})&\text{ if } y=\downn{m_1}{m_2}\text{ for some }\\ 
&m_1,m_2\in M\text{ and }\kappa(m_1)^{(1)}\neq \kappa(m_2)^{(1)}\\
N&\text{ if }y=\downn{m_1}{m_2}\text{ for some}\\
&m_1,m_2\in M \text{ and } \kappa(m_1)^{(1)}=\kappa(m_2)^{(1)}\\
\min\set{ \underset{x \in M\backslash M_a}{\sum}f(\downn{x}{m}): a \in \mathfrak{U}_m }& \text{ if } y=\downnn{w}{m} \text{ for some } m \in M\text{ and }\\
&\text{ and }\mathfrak{U}_m \neq \set{\textbf{0}}\\
1&\text{ if } y=\downnn{w}{m} \text{ for some } m \in M\text{ and }\\
&\text{ and } \mathfrak{U}_m=\set{\textbf{0}}
\end{cases}
$$, where $U$ is the collection of all variables of $\Omega$. The proof is completed by showing that $f$ is a non-negative integral solution of $\Omega$. Since ${\rm CFG(G,c_0)}$ is a ASM, it follows easily that for any two distinct elements $m_1,m_2\in M$, if $\downn{m_1}{m_2}$ and $\downn{m_2}{m_1}$ both are in $U$ then $f(\downn{m_1}{m_2})=f(\downn{m_2}{m_1})$.  It remains to show that for each $m \in M$, $f$ satisfies $\mathfrak{E}(m)$. If $\mathfrak{U}_m=\set{\textbf{0}}$ then the assertion follows easily. If $\mathfrak{U}_m\neq \set{\textbf{0}}$, it is straightforward to verify that $f(\downnn{w}{m})\leq \underset{x\in M\backslash M_a}{\sum} f(\downn{x}{m})$ for any $a \in \mathfrak{U}_m$. We are left with proving $\underset{x \in M\backslash M_a}{\sum}f(\downn{x}{m})<f(\downnn{w}{m})$ for any $a\in \mathfrak{L}_m$. For this purpose, we show that $\underset{x \in M\backslash M_a}{\sum}f(\downn{x}{m})<\underset{x \in M\backslash M_b}{\sum}f(\downn{x}{m})$ for any $b\in \mathfrak{U}_m$. We have
\begin{align}
%\label{fourth equa}
\notag &\underset{x\in M\backslash  M_a}{\sum}f(\downn{x}{m})=\\
\notag &=\underset{\begin{array}{c} x \in M\backslash  M_a \\  \kappa(x)^{(1)}\neq \kappa(m)^{(1)} \end{array}}{\sum}f(\downn{x}{m})+\underset{\begin{array}{c} x \in M\backslash M_a \\ \kappa(x)^{(1)}=\kappa(m)^{(1)}  \end{array}}{\sum}f(\downn{x}{m})=\\
\notag &=\underset{ \begin{array}{c} x \in M\backslash M_a \\ \kappa(x)^{(1)}\neq \kappa(m)^{(1)}  \end{array}  }{\sum}E(\kappa(x)^{(1)},\kappa(m)^{(1)})+Q\times N
\end{align}
where $Q=| \set{(v,n)\in \kappa(M\backslash M_a): v=\kappa(m)^{(1)}} |$. There are two possibilities
\begin{itemize}
  \item[a. ]$Q=\kappa(m)^{(2)}-1$. It follows from \eqref{first equa} and $|\set{(v,n)\in \kappa(M\backslash M_b):v=\kappa(m)^{(1)}}|=\kappa(m)^{(2)}-1$ that 
\begin{align*}
\underset{ \begin{array}{c} x \in M\backslash M_a \\ \kappa(x)^{(1)}\neq \kappa(m)^{(1)}  \end{array}  }{\sum}E(\kappa(x)^{(1)},\kappa(m)^{(1)})&+Q\times N< \underset{\begin{array}{c}x\in M\backslash M_b \\ \kappa(x)^{(1)}\neq \kappa(m)^{(1)}  \end{array}}{\sum}E(\kappa(x)^{(1)},\kappa(m)^{(1)})+\\
&+(\kappa(m)^{(2)}-1)\times N=\\
&=\underset{
\begin{array}{c}
x \in M\backslash M_b
\end{array}
}{\sum} f(\downn{x}{m})
\end{align*}
\item[b. ] $Q<\kappa(m)^{(2)}-1$. It follows from the definition of $N$ that 
\begin{align*}
\underset{ \begin{array}{c} x \in M\backslash M_a \\ \kappa(x)^{(1)}\neq \kappa(m)^{(1)}  \end{array}  }{\sum}&E(\kappa(x)^{(1)},\kappa(m)^{(1)})+Q\times N<N+Q\times N\leq (\kappa(m)^{(2)}-1)\times N\leq \\
&\leq  \underset{\begin{array}{c}x\in M\backslash M_b \\ \kappa(x)^{(1)}\neq \kappa(m)^{(1)}  \end{array}}{\sum}E(\kappa(x)^{(1)},\kappa(m)^{(1)})+(\kappa(m)^{(2)}-1)\times N=\\
&=\underset{
\begin{array}{c}
x \in M\backslash M_b
\end{array}}{\sum} f(\downn{x}{m})
\end{align*}
\end{itemize}  
Therefore, $f$ is a non-negative integral solution of $\Omega$.
\end{proof}
As Lemma \ref{find a solution}, the problem of finding a non-negative integral solution of $\Omega$ is solvable in polynomial time.
\begin{lem}
\label{find a solution for Omega}
Given $\Omega$, we can decide if it has a non-negative integral solution, and if so, find one, in polynomial time.
\end{lem}
\begin{proof}
Clearly the corresponding problem on $\mathbb{R}$ is solvable in polynomial time by using the known algorithms for linear programming. If $\Omega$ has no non-negative real solution then $\Omega$ has no non-negative integral solution. Otherwise let $f':U\to \mathbb{N}$ be a non-negative real solution of $\Omega$, where $U$ denotes the set of variables of $\Omega$. Let $f:U\to \mathbb{N}$ be given by $f(e_{m_1,m_2})=\lfloor 2 |M|f'(e_{m_1,m_2})\rfloor$ if $e_{m_1,m_2} \in U$, and $f(w_m)=min\set{\underset{x \in M\backslash M_a}{\sum}f(e_{x,m}): a \in \mathfrak{U}_m}$. Now we can use the same arguments as in Lemma \ref{find a solution} to argue that $f$ is a non-negative integral solution of $\Omega$. This completes the proof.
\end{proof}

Lemma \ref{find a solution for Omega} implies a polynomial time algorithm for the problem of determining whether a given lattice is in ${\rm L(ASM)}$, and construct a corresponding CFG if there exists one. We again use the Karmarkar's algorithm for finding a non-negative integral solution of $\Omega$. The number of variables of $\Omega$ is bounded by $O(|M|^2)$ and the number of bits, which are input to the algorithms for linear programming to find a non-negative integral solution of $\Omega$, is bounded by $O\left( {|M|}^3\times |X| \right)$. Therefore the algorithm can be implemented to run in $O(|M|^{13}\times |X|^2\times log |X|\times log(log|X|))$ time.\\
\begin{examp}
\label{old example}
Let $L$ be the following lattice
\begin{center}
%the old one is image 29
\includegraphics[bb=11 6 366 209,width=3.1in,height=1.77in,keepaspectratio]{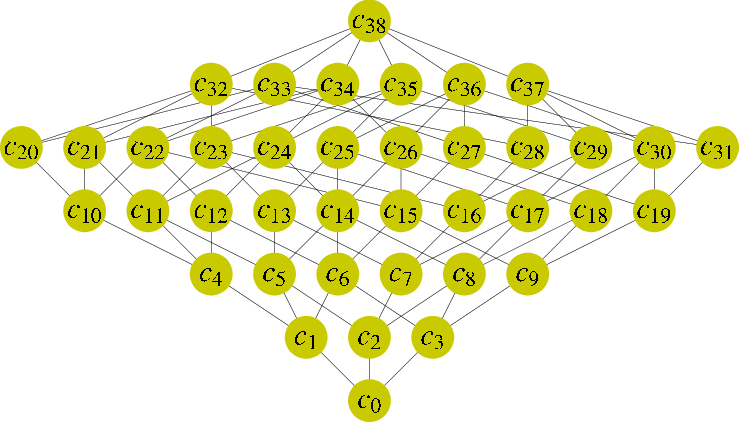}
\end{center}
This latttice was presented in \cite{Mag03} as an example of showing the class of lattices induced by ASM is strictly included in the class of lattices induced by CFG. We again present it here as an application of Theorem \ref{necessary and sufficient condition for ASM}.  The system $\Omega$ of linear inequalities is
$
\{
\downnn{w}{c_{32}}\geq 1,\downnn{w}{c_{33}}\geq 1,\downnn{w}{c_{37}}\geq 1,
\downnn{w}{c_{34}}-(\downn{c_{32}}{c_{34}}+\downn{c_{36}}{c_{34}}+\downn{c_{37}}{c_{34}})\geq 1,
\downnn{w}{c_{34}}\leq \downn{c_{33}}{c_{34}},
\downnn{w}{c_{34}}\leq \downn{c_{32}}{c_{34}}+\downn{c_{35}}{c_{34}},
\downnn{w}{c_{34}}\leq  \downn{c_{35}}{c_{34}}+\downn{c_{36}}{c_{34}}+\downn{c_{37}}{c_{34}},
\downnn{w}{c_{35}}-(\downn{c_{33}}{c_{35}}+\downn{c_{34}}{c_{35}}+\downn{c_{37}}{c_{35}})\geq 1,\downnn{w}{c_{35}}\leq \downn{c_{32}}{c_{35}},
\downnn{w}{c_{35}}\leq \downn{c_{33}}{c_{35}}+\downn{c_{34}}{c_{35}}+\downn{c_{36}}{c_{35}},
\downnn{w}{c_{35}}\leq \downn{c_{36}}{c_{35}}+\downn{c_{37}}{c_{35}},
\downnn{w}{c_{36}}-(\downn{c_{32}}{c_{36}}+\downn{c_{33}}{c_{36}}+\downn{c_{35}}{c_{36}})\geq 1,
\downnn{w}{c_{36}}\leq \downn{c_{33}}{c_{36}}+\downn{c_{34}}{c_{36}},
\downnn{w}{c_{36}}\leq \downn{c_{32}}{c_{36}}+\downn{c_{34}}{c_{36}}+\downn{c_{35}}{c_{36}},
\downnn{w}{c_{36}}\leq \downn{c_{37}}{c_{36}},
\downn{c_{34}}{c_{35}}=\downn{c_{35}}{c_{34}},
\downn{c_{34}}{c_{36}}=\downn{c_{36}}{c_{34}},
\downn{c_{35}}{c_{36}}=\downn{c_{36}}{c_{35}},
\downn{c_{34}}{c_{36}}=\downn{c_{36}}{c_{34}}
\}
$\\
Using the algorithms for linear programming we know that the system has no non-negative  solution, therefore has no non-negative integral solution. Therefore the lattice is not in ${\rm L(ASM)}$.
\end{examp}
\begin{examp}
\label{smaller example}
The game with the initial configuration presented in Figure \ref{fig:image30} generates the lattice presented in Figure \ref{fig:image31}. It is an example which is smaller than one presented in \cite{Mag03}. 
\begin{figure}[h]
\centering
%\parbox{1.2in}{
%\includegraphics[bb=15 6 242 251,width=1.23in,height=1.33in,keepaspectratio]{image30}
%\caption{Initial configuration}
%\label{fig:image30}
%}
%\qquad \qquad
%\parbox{1.2in}{
%\includegraphics[bb=11 5 321 243,width=2.3in,height=1.77in,keepaspectratio]{image31}
%\caption{Generating lattice}
%\label{fig:image31}
%}
\subfloat[Initial configuration]{\label{fig:image30}\includegraphics[bb=15 6 242 251,width=1.23in,height=1.33in,keepaspectratio]{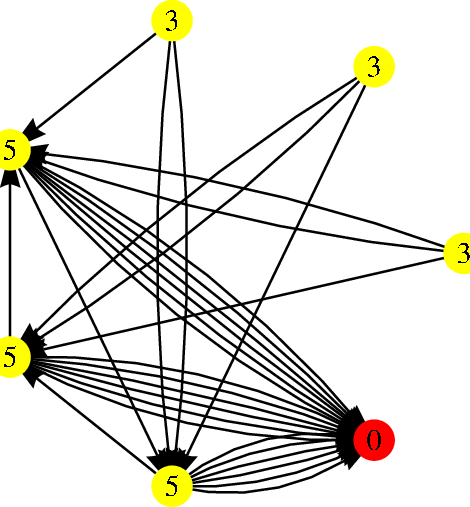}}
\qquad \qquad
\subfloat[Generating lattice]{\label{fig:image31}\includegraphics[bb=11 5 321 243,width=2.3in,height=1.77in,keepaspectratio]{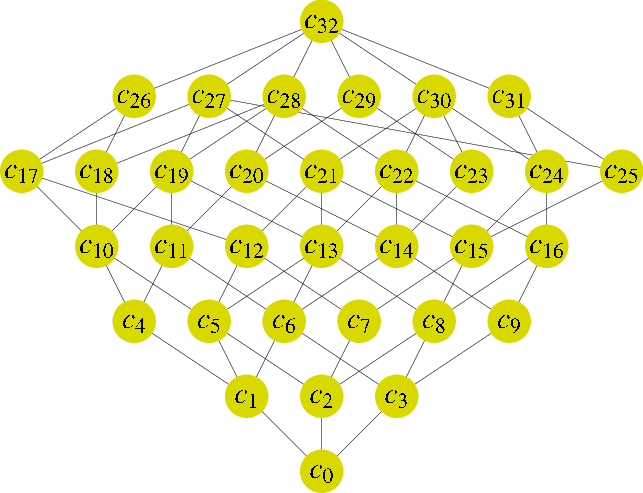}}
\caption{Smaller example}
\label{fig:image3031}
\end{figure}
%\begin{center}
%\includegraphics[bb=15 6 242 251,width=1.23in,height=1.33in,keepaspectratio]{image30}
%\end{center}
%generates the lattice
%\begin{center}
% \includegraphics[bb=11 5 321 243,width=2.3in,height=1.77in,keepaspectratio]{image31}
%\end{center}
\end{examp}
Note that the two lattices in Example \ref{old example} and Example \ref{smaller example} are generated only by simple CFGs. It is useful to give a sufficient condition for such lattices. The following proposition shows such a condition
\begin{prop} 
Let $H$ be the undirected simple graph whose vertices are $M$ and edges are defined by $(m_1,m_2) \in E(H)$ if there are $x_1,x_2,x_3 \in X$ such that $x_1 \prec x_2, x_1 \prec x_3$, $\mathfrak{m}(x_1,x_2)=m_1$ and $\mathfrak{m}(x_1,x_3)=m_2$. If $L \in {\rm L(CFG)}$ and $H$ is a complete graph then $L$ is generated only by simple CFGs.
\end{prop}
\begin{proof}
We assume otherwise that $L$ is generated by a non-simple CFG, say ${\rm CFG(G,c_0)}$. Let $\kappa$ be the map defined in Lemma \ref{generalization of lemma of relation}. Since the game is not simple, there is a vertex $v$ of $G$  such that $v$ is fired more than once during an execution from $c_0$ to the fixed point. Therefore there are two distinct meet-irreducibles $m_1$ and $m_2$ in $M$ such that $\kappa(m_1)=(v,1)$ and $\kappa(m_2)=(v,2)$. Since $H$ is complete, there are three distinct configurations $x_1,x_2$ and $x_3$ such that $x_1\prec x_2,x_1\prec x_3,\mathfrak{m}(x_1,x_2)=m_1$ and $\mathfrak{m}(x_1,x_3)=m_2$. By the definition of $\kappa$ we have $x_1\overset{v}{\to}x_2$ and $x_1\overset{v}{\to}x_3$. It implies that $x_2=x_3$, a contradiction.
\end{proof}

In \cite{MVP01}, the authors proved that a general CFG  is always equivalent to a simple CFG. An arising question is that whether a ASM is equivalent to a simple ASM. The idea from the proof in \cite{MVP01} does not seem to be applicable to this model, whereas the answer follows easily from the proofs of Lemma \ref{sufficient condition for L(ASM)} and Theorem \ref{necessary and sufficient condition for ASM}
\begin{prop}
Any ASM is equivalent to a simple  ASM.
\end{prop}
\begin{proof}
Assume that $L\in {\rm L(ASM)}$. By Theorem \ref{necessary and sufficient condition for ASM} $\Omega$ has non-negative integral solutions. We consider the CFG that is constructed as in the proof of Lemma \ref{sufficient condition for L(ASM)}. It is a simple ASM and generates $L$. This completes the proof.
\end{proof}
\section{CFGs on acyclic graphs}
\label{fifth section}
In \cite{Mag03} the author gave a strong relation between ASM and the simple CFGs on acyclic graphs (directed acyclic graphs). The author pointed out that a simple CFG on an acyclic graph is equivalent to a ASM. In this section we study CFGs on acyclic graphs that are not necessarily simple. We show that each CFG on an acyclic graph is equivalent to a simple CFG on an acyclic graph. As a corollary, every lattice generated by a CFG on an acyclic graph is in ${\rm L(ASM)}$. We also give a necessary and sufficient criterion for lattices generated by CFGs on acyclic graphs.

\begin{lem}
\label{CFGs on acyclic graphs}
If $L$ is generated by a CFG on an acyclic graph then $\mathcal{G}$ is acyclic, where $\mathcal{G}$ is the simple directed graph whose vertices are $M$ and edges are defined by $(m_1,m_2)\in E(\mathcal{G})$ if and only if $m_1 \in \underset{a \in \mathfrak{U}_{m_2}}{\bigcup}(M\backslash M_a)$.
\end{lem}
\begin{proof}
Let ${\rm CFG(G,c_0)}$ be a CFG which generates $L$, where $G$ is an directed acyclic graph. Let $\kappa:M\to V(G)\times \mathbb{N}$ be the map which is defined in Lemma \ref{generalization of lemma of relation}. For each $v \in V(G)$, $\downvertices{v}$ denotes the collection of vertices $v'$ of $G$ such that  there is a directed path from $v'$ to $v$ in $G$ ($v'$ could be equal to $v$). A sequence $(v_1,v_2,\cdots,v_k)\in V(G)^{k}$ is called a \emph{valid firing sequence} if there is an execution $c_0\overset{v_1}{\to}{c_1}\overset{v_2}{\to}{c_2}\overset{v_3}{\to}\cdots \overset{v_k}{\to}{c_k}$ in the game. Note that if such an execution exists then it is defined uniquely. 

We claim that for each $m \in M$ and each $a \in \mathfrak{U}_m$, we have $\set{\kappa(x)^{(1)}: x \in M\backslash M_a}\subseteq \downvertices{\kappa(m)^{(1)}}$. Indeed, let $\textbf{0}=c_0\overset{v_1}{\to}{c_1}\overset{v_2}{\to}c_2\overset{v_3}{\to}\cdots \overset{v_k}{\to}c_k=a\overset{\kappa(m)^{(1)}}{\to}j_m$ be an execution, where $j_m$ is the configuration which is obtained by firing $\kappa(m)^{(1)}$ in $a$. It is clear that $(v_1,v_2,\cdots,v_k,\kappa(m)^{(1)})$ is a valid firing sequence. Since firing of the vertices not in $\downvertices{\kappa(m)^{(1)}}$ does not affect the firability of the remaining vertices, by removing all vertices not in $\downvertices{\kappa(m)^{(1)}}$ of the sequence, we get the sequence $(v_{i_1},v_{i_2},\cdots,v_{i_{k'}},\kappa(m)^{(1)})$ which remains a valid firing sequence. There exists an execution $c_0=d_0\overset{v_{i_1}}{\to}d_1\overset{v_{i_2}}{\to}d_2\overset{v_{i_3}}{\to}\cdots \overset{v_{i_{k'}}}{\to}d_{k'}=a'\overset{\kappa(m)^{(1)}}{\to}j_m'$ in the game. Since the number of occurrences of $\kappa(m)^{(1)}$ in $(v_{i_1},v_{i_2},\cdots,v_{i_{k'}},\kappa(m)^{(1)})$ is the same as the one in $(v_1,v_2,\cdots,v_k,\kappa(m)^{(1)})$, it follows that $\kappa(\mathfrak{m}(a,j_m))=\kappa(\mathfrak{m}(a',j_m'))$, thus $\mathfrak{m}(a,j_m)=\mathfrak{m}(a',j_m')=m$. Clearly we have $\oneshotvector{a'}\leq \oneshotvector{a}$, therefore $a'\leq a$. It follows from the definition of  $\mathfrak{U}_m$ that $a'=a$, hence $\set{\kappa(x)^{(1)}: x\in M\backslash M_a}=\set{\kappa(\mathfrak{m}(c_i,c_{i+1}))^{(1)}:0 \leq i \leq k-1}=\set{v_i : 1 \leq i \leq k} \subseteq \downvertices{\kappa(m)^{(1)}}$.

Our next claim is that for each $m \in M$ and each $a \in \mathfrak{U}_m$, if $(\kappa(m)^{(1)},n)\in \kappa(M\backslash M_a)$ then $n < \kappa(m)^{(2)}$. Indeed, let $\textbf{0}=c_0\overset{v_1}{\to}c_1\overset{v_2}{\to}c_2\overset{v_3}{\to}\cdots \overset{v_k}{\to}c_k=a\overset{\kappa(m)^{(1)}}{\to} j_m$ be an execution. Since $\set{\mathfrak{m}(c_i,c_{i+1}):0\leq i \leq k-1}=M\backslash M_a$, there is $p\in \set{0,1,2,\cdots, k-1}$ such that $\kappa(\mathfrak{m}(c_p,c_{p+1}))=(\kappa(m)^{(1)},n)$. Thus $v_{p+1}=\kappa(m)^{(1)}$, and $n=\oneshotvector{c_{p+1}}(v_{p+1})\leq \oneshotvector{a}(v_{p+1})=\oneshotvector{j_m}(v_{p+1})-1<\oneshotvector{j_m}(v_{p+1})=\kappa(m)^{(2)}$. The claim follows.

Since $G$ is an acyclic graph, there exists a function $h:V(G)\to \mathbb{N}$ such that if $(v_1,v_2)\in E(G)$ then $h(v_1)<h(v_2)$. Let $N=\max \set{\shotvector{\textbf{1}}{v}: v\in V(G)}$. We define $h':M\to \mathbb{N}$ by $h'(m)=N\times h(\kappa(m)^{(1)})+\kappa(m)^{(2)}$. To prove $\mathcal{G}$ is acyclic, it suffices to show that for every $(m_1,m_2) \in \mathcal{G}$, we have $h'(m_1)<h'(m_2)$. From the definition of $\mathcal{G}$, there exists $a \in \mathfrak{U}_{m_2}$ such that $m_1 \in M\backslash M_a$. There are two possibilities
\begin{itemize}
  \item[a. ] $\kappa(m_1)^{(1)}=\kappa(m_2)^{(1)}$.  It follows from the second claim that $h'(m_1)=N \times h(\kappa(m_1)^{(1)})+\kappa(m_1)^{(2)}<N \times h(\kappa(m_2)^{(1)})+\kappa(m_2)^{(2)}=h'(m_2)$. 
  \item[b. ] $\kappa(m_1)^{(1)}\neq \kappa(m_2)^{(1)}$. It follows from the first claim that $\kappa(m_1)^{(1)}\in \downvertices{\kappa(m_2)^{(1)}}$.  Therefore $h'(m_1)=N \times h(\kappa(m_1)^{(1)})+\kappa(m_1)^{(2)}\leq N \times h(\kappa(m_2)^{(1)})+\kappa(m_2)^{(2)}-N+(\kappa(m_1)^{(2)}-\kappa(m_2)^{(2)})<N \times h(\kappa(m_2)^{(1)})+\kappa(m_2)^{(2)}=h'(m_2)$.
\end{itemize}  
\end{proof}
We recall a result in \cite{Mag03}
\begin{theo}
\label{theorem of Magnien}\cite{Mag03} 
Let $C$ be a simple CFG on an acyclic graph $G$. Then $C$ is equivalent to a ASM.
\end{theo}
Here, our main result of this section
\begin{theo}
\label{CFGs on acyclic graph including in ASM}
Any CFG on an acyclic graph is equivalent to a simple CFG on an acyclic graph, therefore equivalent to a ASM.
\end{theo}
\begin{proof}
Let ${\rm CFG(G,c_0)}$ be a CFG such that $G$ is an acyclic graph, and let $L$ denote ${\rm CFG(G,c_0)}$. By Theorem \ref{the condition of lattices induced by CFG} that for each $m \in M$, $\mathcal{E}(m)$ has non-negative integral solutions. Let $U_m$ be the collection of all variables of $\mathcal{E}(m)$ and ${f'}_m:U_m\to \mathbb{N}$ be a non-negative integral solution of $\mathcal{E}(m)$. The function $f_m:U_m\to \mathbb{N}$ defined by
$$
f_m(y)=\begin{cases}
{f'}_m(y)&\text{ if } y=w\text{ or } y\in \set{\down{x}: x\in \underset{a \in \mathfrak{U}_m}{\bigcup} (M\backslash M_a)}\\
0&otherwise
\end{cases}
$$
is also a non-negative integral solution of $\mathcal{E}(m)$. By using solutions $f_m$,  it follows from the construction of the CFG in the proof of Theorem \ref{the condition of lattices induced by CFG} that $L$ is generated by a simple CFG on a graph, say $G'$, such that $V(G')=M\cup \set{s}$ and if $(v_1,v_2) \in E(G')$ then $v_2=s$ or $(v_1,v_2) \in E(\mathcal{G})$, where $\mathcal{G}$ is the graph that is defined as in Lemma \ref{CFGs on acyclic graphs}. It follows directly from Lemma \ref{CFGs on acyclic graphs} that  $\mathcal{G}$ is acyclic, so is $G'$. Theorem \ref{theorem of Magnien} implies that ${\rm CFG(G,c_0)}$ is equivalent to a ASM.
\end{proof}
Using Lemma \ref{CFGs on acyclic graphs} and a similar argument as in the proof of Theorem \ref{CFGs on acyclic graph including in ASM} we obtain a necessary and sufficient criterion for the class of lattices generated by CFGs on acyclic graphs
\begin{coro}
Let $L \in {\rm L(CFG)}$. Then $L$ is generated by a CFG on an acyclic graph if and only if $\mathcal{G}$ is acyclic.
\end{coro}
\begin{proof}
The necessary condition is proved by Lemma \ref{CFGs on acyclic graphs}. It remains to prove that the sufficient condition also hold. Let $\mathcal{G}$ be given as in Lemma \ref{CFGs on acyclic graphs}. By Theorem \ref{the condition of lattices induced by CFG} that for each $m \in M$, $\mathcal{E}(m)$ has non-negative integral solutions. We define non-negative integral solutions $f'_m$, $f_m$ for each $m \in M$, and a simple CFG on $G'$ that generates $L$  as in Theorem \ref{CFGs on acyclic graph including in ASM}. Since $\mathcal{G}$ is acyclic  and since if $(v_1,v_2) \in E(G')$ then $v_2=s$ or $(v_1,v_2) \in E(\mathcal{G})$, it follows that $G'$ is acyclic. This completes the proof.
\end{proof}
Let ${\rm L(ACFG)}$ denote the class of lattices generated by CFGs on acyclic graphs. Theorem \ref{CFGs on acyclic graph including in ASM} implies that ${\rm L(ACFG)}\subseteq {\rm L(ASM)}$. We consider the lattice shown in Figure \ref{fig:image16}. In this case, $\mathcal{G}$ is presented by the following figure
\begin{center}
\includegraphics[bb=0 2 150 155,width=1.2in,height=1.22in,keepaspectratio]{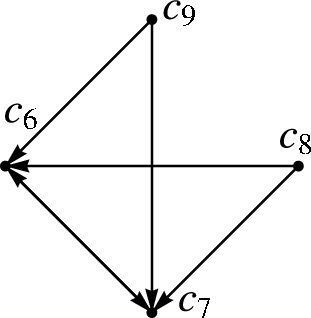}
\end{center}
$\mathcal{G}$ is not acyclic, therefore the lattice is not in ${\rm L(ACFG)}$. From Example \ref{second example ?}, the lattice is in ${\rm L(ASM)}$. It implies that ${\rm L(ACFG)}\subsetneq {\rm L(ASM)}$. Furthermore the lattice presented in Figure \ref{fig:image37}
%\begin{figure}[tbp] % float placement: (h)ere, page (t)op, page (b)ottom, other (p)age
%  \centering
%  % file name: E:/PhD/ULDlattices2/Images/image37.pdf
%  \includegraphics[bb=4 4 107 151,width=1.24in,height=1.77in,keepaspectratio]{image37}
%  \caption{A lattice in $L(ACFG)\backslash D$}
%  \label{fig:image37}
%\end{figure}
\begin{figure}[h]
\centering
\includegraphics[bb=4 4 107 151,width=1.24in,height=1.77in,keepaspectratio]{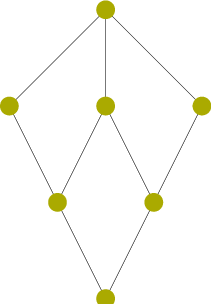}
\caption{A lattice in ${\rm L(ACFG)}\backslash D$}
\label{fig:image37}
\end{figure}
is generated by a CFG on acyclic graph but not a distributive lattice. Thus $D \subsetneq {\rm L(ACFG)}$.
\section{Conclusion and perspectives}
\label{sixth section}
In this paper we have studied the properties of three classes of lattices generated by CFGs, that are ${\rm L(CFG)}$, ${\rm L(ASM)}$ and ${\rm L(ACFG)}$. On algorithmic aspect we give a necessary and sufficient criterion for each studied class. These criteria provide the polynomial-time algorithms for determining which class of lattices a given ULD lattice belongs to. A relation between those classes of lattices is also pointed out by showing that ${\rm L(ACFG)}$ is situated strictly between the class of distributive lattices and ${\rm L(ASM)}$. In other word, we obtain a finer chain of the studied classes of lattices, that is
$$
D \subsetneq {\rm L(ACFG)} \subsetneq {\rm L(ASM)}\subsetneq {\rm L(CFG)}={\rm L(MCFG)}\subsetneq  ULD,
$$
where ${\rm L(MCFG)}$ is the class of lattices generated by MCFGs (Multating Chip Firing Game \cite{B97,H99,Mag03}).

It is interesting to investigate CFGs defined on the classes  of graphs that are studied widely in literature, for example the class of Eulerian directed graphs. This class is a close extension of the class of undirected graphs. Recall that a graph $G$ is \emph{Eulerian} if it is connected and for each vertex of $G$ its out-degree and in-degree are equal. We define a CFG on an Eulerian graph $G$ as follows. We fix a vertex $s$ of $G$ which will play a role as the sink of the game. Then we remove all out-edges of $s$. The resulting graph $G'$ remains a connected graph and has no closed component. The game is defined on this graph. Let $L(ECFG)$ denote the class of lattices generated by CFGs on Eulerian graphs. It is clear that ${\rm L(ASM)}\subseteq {\rm L(ECFG)}\subseteq {\rm L(CFG)}$. The problem of determining which inclusion is strict remains to be done.  

It turns out to be interesting that a CFG defined on each studied class of graphs is equivalent to a simple CFG which again is defined on this class. This property implies that to study the lattices generated by CFGs defined on these classes of graphs, it is sufficient to study simple CFGs. As we saw in this paper this property is proved on different classes of graphs with different techniques. Thus it is not easy to know whether this property holds for other classes of graphs. In particular we still do not know whether this property holds for the Eulerian graphs. A  characterization of classes of graphs having this property remains to be done.

Finally, we are also interested in the following computational problem: Given a graph $G$ and a ULD lattice $L$, is $L$ generated by a CFG on $G$?

So now, we have the practical criteria for the classes of lattices generated by CFGs defined on three classes of graphs which are studied widely in literature, they are acyclic graphs, undirected graphs, and directed graphs. We believe that our method presented here is not only applicable to these classes but also applicable to many other classes of graphs on which CFGs are defined.
\text{}\\

\textbf{Acknowledgement.} We would like to thank D. Dhar for the comments on the term ``Abelian Sandpile model'' that was used in the first version (in arXiv and also in journal). The more informations about  this term have been added to this version.
\pagebreak
%\begin{small}
%\bibliographystyle{plain}
%\bibliography{BIBLIO}
%\end{small}

Trung Van Pham\\
Department of Mathematics of Computer Science\\
Vietnam Institute of Mathematics\\
18 Hoang Quoc Viet road, Cau Giay district, Hanoi, Vietnam\\
Email: pvtrung@math.ac.vn\\
\text{}\\
Thi Ha Duong Phan\\
Department of Mathematics of Computer Science\\
Vietnam Institute of Mathematics\\
18 Hoang Quoc Viet road, Cau Giay district, Hanoi, Vietnam\\
Email: phanhaduong@math.ac.vn
\end{document}